\newif\ifready\readyfalse
\pdfoutput=1

\documentclass{article}
\usepackage[utf8]{inputenc}
\usepackage{graphicx}
\usepackage{booktabs}
\usepackage[top=1in,left=1in,right=1in,bottom=1in]{geometry}
\usepackage{amsthm}
\usepackage{amsmath}

\usepackage{amssymb}
\usepackage[colorlinks=true,linkcolor=blue,citecolor=blue,allcolors=blue]{hyperref}
\hypersetup{
    colorlinks=true,
    linkcolor=blue,
    filecolor=blue,
    urlcolor=blue,
    allcolors=blue
}
\usepackage{relsize}
\usepackage{algorithmicx}
\usepackage{algorithm}
\usepackage[noend]{algpseudocode}
\usepackage{multicol}
\usepackage[most]{tcolorbox}
\usepackage{color, soul}
\usepackage{pgfplots}
\usepackage{subcaption}
\pgfplotsset{
    compat=1.3,
    legend image code/.code={
        \draw [#1] (0cm,-0.1cm) rectangle (0.6cm,0.1cm);
    },
}
\usepackage{parskip}

\usepackage{thmtools}
\usepackage{thm-restate}
\usepackage{enumerate}
\usepackage{mathtools}
\usepackage{xspace}

\usepackage[capitalize,nameinlink]{cleveref}

\theoremstyle{plain}
\newtheorem{theorem}{Theorem}[section]

\newtheorem{lemma}[theorem]{Lemma}
\newtheorem{corollary}[theorem]{Corollary}
\newtheorem{definition}[theorem]{Definition}

\newtheorem{observation}[theorem]{Observation}
\newtheorem{example}{Example}
\newtheorem*{remark}{Remark}

\newcommand{\defn}[1]{\textbf{\textit{#1}}}

\crefname{theorem}{Theorem}{Theorems}
\Crefname{lemma}{Lemma}{Lemmas}
\Crefname{claim}{Claim}{Claims}
\Crefname{observation}{Observation}{Observations}
\Crefname{algorithm}{Algorithm}{Algorithms}
\Crefname{myalgctr}{Algorithm}{Algorithms}
\Crefname{challenge}{Challenge}{Challenges}

\algrenewcommand\algorithmicindent{1em}%

\newcommand{\alg}{\mathcal{A}}

\newcommand{\tO}{\widetilde{O}}

\DeclarePairedDelimiter\floor{\lfloor}{\rfloor}
\usepackage{xcolor}

\algblock{ParFor}{EndParFor}
\algblock{Input}{EndInput}
\algblock{Output}{EndOutput}
\algblock{ReduceAdd}{EndReduceAdd}
\algblock{PhaseOne}{EndPhaseOne}
\algblock{PhaseTwo}{EndPhaseTwo}
\algblock{Maintain}{EndMaintain}

\algnewcommand\algorithmicparfor{\textbf{parfor}}
\algnewcommand\algorithmicinput{\textbf{Input:}}
\algnewcommand\algorithmicoutput{\textbf{Output:}}
\algnewcommand\algorithmicphaseone{\textbf{Phase 1:}}
\algnewcommand\algorithmicphasetwo{\textbf{Phase 2:}}
\algnewcommand\algorithmicmaintain{\textbf{Maintain:}}
\algnewcommand\algorithmicreduceadd{\textbf{ReduceAdd}}
\algnewcommand\algorithmicpardo{\textbf{do}}
\algnewcommand\algorithmicendparfor{\textbf{end\ input}}

\algnewcommand{\parState}[1]{\State%
    \parbox[t]{\dimexpr\linewidth-\algmargin}{\strut #1\strut}}

\algnewcommand{\parComment}[1]{\Comment%
    \parbox[t]{\dimexpr\linewidth-\algmargin}{\strut #1\strut}}

\algrenewtext{ParFor}[1]{\algorithmicparfor\ #1\ \algorithmicpardo}
\algrenewtext{Input}[1]{\algorithmicinput\ #1}
\algrenewtext{Output}[1]{\algorithmicoutput\ #1}
\algrenewtext{ReduceAdd}[2]{#1 $\leftarrow$ \algorithmicreduceadd(#2)}
\algrenewtext{PhaseOne}[1]{\algorithmicphaseone\ #1}
\algrenewtext{PhaseTwo}[1]{\algorithmicphasetwo\ #1}
\algrenewtext{Maintain}[1]{\algorithmicmaintain\ #1}
\algtext*{EndInput}
\algtext*{EndOutput}
\algtext*{EndIf}
\algtext*{EndFor}
\algtext*{EndWhile}
\algtext*{EndParFor}
\algtext*{EndReduceAdd}
\algtext*{EndPhaseOne}
\algtext*{EndPhaseTwo}
\algtext*{EndMaintain}

\DeclareMathOperator{\poly}{poly}

\definecolor{mygreen}{RGB}{20,140,80}
\definecolor{linkcolor}{RGB}{0,0,230}
\definecolor{mylightgray}{RGB}{230,230,230}
\definecolor{verylightgray}{RGB}{245,245,245}

\newcounter{myalgctr}

\newtcolorbox{OuterBox}[1][]{%
    breakable,
    enhanced,
    frame hidden,
    interior hidden,
    left=-5pt,
    right=-5pt,
    top=-5pt,
    float=p,
    boxsep=0pt,
    arc=0pt
#1}%

\newtcolorbox{InnerBox}[1][]{%
    enforce breakable,
    enhanced,
    colback=gray,
    colframe=white,
#1}%

\newenvironment{tbox}{
\vspace{0.2cm}
\begin{tcolorbox}[width=\columnwidth,
                  enhanced,
                  boxsep=2pt,
                  left=1pt,
                  right=1pt,
                  top=4pt,
                  boxrule=1pt,
                  arc=0pt,
                  colback=white,
                  colframe=black,
	              breakable
                  ]%
}{
\end{tcolorbox}
}

\newcommand{\tboxhrule}[0]{\vspace{0.1cm} {\color{black} \hrule} \vspace{0.2cm}}

\newenvironment{titledtbox}[1]{\begin{tbox}#1 \tboxhrule}{\end{tbox}}

\newcommand{\batch}{\mathcal{B}}

\newcommand{\eps}{\varepsilon}

\newcommand{\prob}{\mathbf{P}}
\newcommand{\geom}{\mathsf{Geom}}

\newcommand{\hatT}{\widehat{T}}
\newcommand{\delvec}{\mathbf{d}}
\newcommand{\insvec}{\mathbf{i}}
\newcommand{\cS}{\mathcal{S}}

\newcommand{\fact}{\log^2 n}

\newcommand{\draw}{\geom(\exp(\eps/\fact))}

\newcommand{\expect}{\mathbf{E}}

\newcommand{\initial}{\eps^2}

\newcommand{\trep}{t_{\text{report}}}
\newcommand{\tdel}{t_{\text{delete}}}
\newcommand{\maxe}{error_{\text{max}}}
\newcommand{\mT}{\mathcal{T}}
\newcommand{\predict}{\text{predict}}
\newcommand{\matchP}{\overline{P}}
\newcommand{\realvec}{\mathbf{r}}
\newcommand{\predvec}{\mathbf{p}}
\newcommand{\assignedvec}{\mathbf{a}_0}
\newcommand{\find}{\textsc{Find}\xspace}
\newcommand{\union}{\textsc{Union}\xspace}

\makeatletter
\newcounter{HALG@line}
\renewcommand{\theHALG@line}{\thealgorithm.\arabic{ALG@line}}
\makeatother

\begin{document}

\title{The Predicted-Updates Dynamic Model: Offline, Incremental, and Decremental to Fully Dynamic Transformations}
\date{}
\author{Quanquan C. Liu\thanks{Simons Institute at UC Berkeley: \href{mailto:quanquan@mit.edu}{\texttt{quanquan@mit.edu}}, supported by an Apple Research Fellowship.} ~and Vaidehi Srinivas\thanks{Northwestern University: \href{mailto:vaidehi@u.northwestern.edu}{\texttt{vaidehi@u.northwestern.edu}}, supported by the National Science Foundation (NSF) under grant Nos. CCF-1652491, CCF-1934931 and EECS-2216970.}}

\pagenumbering{gobble}
\clearpage
\maketitle
\sloppy

\begin{abstract}
    The main bottleneck in designing efficient dynamic algorithms is the unknown nature of the update sequence.  In particular, there are some problems, like triconnectivity, planar digraph all pairs shortest paths, $k$-edge connectivity, and others, where the separation in runtime between the best offline or partially dynamic solutions and the best fully dynamic solutions is polynomial, sometimes even exponential. 

    In this paper, we formulate the \emph{predicted-updates dynamic model}, one of the first \emph{beyond-worst-case} models for dynamic algorithms, which generalizes a large set of well-studied dynamic models including the offline dynamic, incremental, and decremental models to the fully dynamic setting when given predictions about the update times of the elements. In the most basic form of our model, we receive a set of predicted update times for all of the updates that occur over the event horizon.  We give a novel framework that ``lifts" offline divide-and-conquer algorithms into the fully dynamic setting with little overhead.  Using this, we are able to interpolate between the offline and fully
    dynamic settings; when the $\ell_1$ error of the prediction is linear in the number of updates, we achieve the offline runtime of the algorithm (up to $\poly\log n$ factors).  Provided a fully dynamic \emph{backstop} algorithm, our algorithm will never do worse than the backstop algorithm regardless of the prediction error. Furthermore, our framework achieves a smooth linear trade-off between $\ell_1$ error in the predictions and runtime.  These correspond to the desiderata of consistency, robustness, and graceful degradation of the algorithms-with-predictions literature.  We further extend our techniques to incremental and decremental settings, transforming algorithms in these settings when given predictions of only the deletion and insertion times, respectively.  Our framework is general, and we apply it to obtain improved efficiency bounds over the state-of-the-art dynamic algorithms for a variety of problems, for prediction error of reasonable magnitude.
    
    Our paper models real world settings, in which we often have access to side information that allows us to make coarse predictions about future updates. Our techniques are also of theoretical interest, as they interpolate between the offline/partially dynamic and fully dynamic settings, and provide natural extensions of the \emph{algorithms-with-predictions} paradigm to the dynamic setting.\footnote{The previous posted version of this manuscript focused on incremental to fully dynamic transformations.  The new version merges with a previous unpublished manuscript from April 2023 to include a more general framework including offline to fully dynamic transformations.}
\end{abstract}

\newpage
\tableofcontents
\newpage
\clearpage
\pagenumbering{arabic}

\section{Introduction}

Learning-augmented algorithms and traditional dynamic algorithms have fundamentally similar goals: to efficiently provide
up-to-date solutions to problems in the face of rapidly rapidly changing data. 
While learning-augmented algorithms typically use machine learning methods to take advantage of structure in data, traditional dynamic algorithms consider worst-case adversarial inputs, and algorithms are designed to be 
efficient in spite of such worst-case, potentially unstructured, instances.
In this paper, we combine these two 
paradigms under \defn{one framework} that encompasses the three major settings traditionally studied within dynamic algorithms: 
the \emph{offline dynamic}, \emph{incremental}, and \emph{decremental} settings. Specifically, we show that we can \emph{transform}
broad classes of offline, incremental, and decremental algorithms into \emph{fully dynamic} algorithms provided predictions of the update
times of dynamic events such as insertions and/or deletions. We show that we achieve the same runtimes as the original offline, incremental, 
and decremental algorithms, up to polylogarithmic factors, when the $\ell_1$ error of our predictions is near-linear in the number of updates 
that are performed. Not only does this allow us to achieve (sometimes exponential) improvements in runtime
for certain problems, %
but it also allows us to transfer algorithmic techniques between models in novel ways.

\emph{Algorithms with predictions} or \emph{learning-augmented algorithms} is a very rich area~\cite{MMV17, HIKV18, KBCDP18, Mit18, IVY19, ACEPS20, ADJKR20, BMRS20, BMS20, cao22learning, CGP20, JLLRW20, JPS20, LLMV20, lin2022learning, Roh20, Wei20, DILMV21, DLPV21, EINRSW21, LMRX21, Mit21, VKMK21, CEILMRSWWZ22, CSVZ22, IMMR22, DMVW23} that in recent times has seen many interesting and important results. 
In this paradigm, an algorithm solicits \emph{predictions} from an untrusted source (e.g.\ a machine learning model) to help make decisions.  The goal is to design algorithms that achieve the following three desiderata \cite{LV21}: 
\begin{enumerate}[(1)]
    \item \emph{(Consistency)} If the predictions are of high quality, the algorithm performs much better than a worst-case algorithm. 
    \item \emph{(Competitiveness)} If the predictions are of low quality, the algorithm does not perform any worse than a worst-case algorithm. %
    \item \emph{(Robustness)} The performance of the algorithm degrades gracefully as a function of the prediction error.
\end{enumerate}
Additionally, we want to solicit predictions that can be reasonably obtained in practice. An example 
is link-prediction for dynamic graphs (where edges are 
inserted and deleted)
which has been gaining great interest recently in the machine learning community, with a set of promising results (a small sample
includes~\cite{kumar2019predicting, nguyen2018continuous,
poursafaei2022towards, quach2021dyglip, rossi2020temporal,zhang2018link, wang2021inductive}). 
A detailed overview of the field is given in~\cite{BWCA-AlgosWithPredictions}.

These desiderata present a challenge. While we would like to achieve all three for some reasonable notion of prediction, it is not a priori clear that such a guarantee is even possible.  For some problems and settings, it is indeed not possible, 
and algorithms are designed to trade off these objectives.
Specifically, in the setting of \emph{online algorithms with predictions}, which is closely related to our work, 
we see inherent tradeoffs between consistency and robustness  (e.g. in the the classic problem
of rent-or-buy \cite{KPS18, GP19, WZ20}). Such tradeoffs are unavoidable, as online problems often involve an aspect of ``commitment."  An algorithm must commit to either following the worst-case strategy or following 
the prediction; in either case, the adversary can then choose inputs to make the algorithm lose in either consistency or robustness based on this choice.  Thus, the main research direction for online algorithms with predictions is in quantifying these tradeoffs.  

Another setting that is closely related to our work is that of \emph{warm starts}, where we design algorithms for static problems that take advantage of a predicted solution.  The goal in this setting is to minimize runtime on a static input.  Here, it is usually possible to achieve both consistency and robustness, by simply running the learning-augmented procedure in parallel with a worst-case procedure, and outputting the answer of the algorithm which completes first.  However, 
unlike the online setting where we only need predictions of future events, 
warm starts often require the prediction to be very large 
\cite{DILMV21, CSVZ22, DMVW23}, typically encoding a \emph{solution} to the problem.
In the context of graph problems, such solutions could have size proportional to the size of the graph, and even reading the prediction can take nontrivial time, much less verifying it. 
We provide a more detailed 
comparison of these settings with our work in~\cref{app:related-work}.

Thus, it is particularly interesting that for dynamic algorithms, it is indeed possible to  
achieve the best of online algorithms and warm starts: we only need predictions of future events 
(instead of future solutions) and we \emph{simultaneously achieve all three desiderata} of algorithms with predictions. Our model has the added benefit that the requested predictions are not 
problem-specific, and the same predictions could be used multiple times to solve different unrelated problems (e.g.\ one set of predictions about a dynamic graph can be used to solve a collection of graph problems). %

Our framework expands upon the traditional model used for dynamically changing data.
The main bottleneck of designing traditional \emph{dynamic algorithms} is the unknown nature of the update sequence.
This is demonstrated by the large gap in efficiency between algorithms for the offline and partially dynamic settings 
and the fully dynamic (online) setting for many graph problems (see e.g.,~\cite{abraham2017fully, bernstein2018incremental, baswana2019dynamic, 
bhattacharya2020improved, ChenDWZZ18, chen2020fast, das2022near, EGIN97, forster2021dynamic,
ForsterNG23, Goranci19, GHS19, GIS99, gao2023fully, goldberg1998beyond, goranci2021expander, HKMST12, holm2020worst, JS21,
klein1998fully, PSS17}).
To address this, there has been much
interest in using the techniques from the offline and partially dynamic settings to obtain 
fully dynamic algorithms for certain problems (see e.g.~\cite{abraham2016fully, abboud2014popular, bentley1979decomposable,
baswana2012fully, bentley1980decomposable, dobkin1991maintenance, henzinger1999randomized,
henzinger2016dynamic, henzinger2015unifying, holm2001poly, king1999fully, kopelowitz2016higher, overmars2005dynamization, overmars1979two,
probst2020new, roditty2004dynamic}). 

Following this line of work, Chan~\cite{chan2011three} and, very recently, 
Peng and Rubinstein~\cite{PR23} gave an elegant systematic transformation for \emph{incremental} algorithms to 
fully dynamic algorithms in the \emph{known-deletion} model. %
An incremental algorithm is one that produces answers
after insertions of elements, while a fully dynamic algorithm is one that produces answers after both element insertions 
and deletions. 
The \emph{known-deletion} model is an intermediate model, between incremental and fully dynamic, in which an inserted element arrives tagged with the precise timestamp of when it will be deleted.\footnote{\cite{PR23} refer to this model as the \emph{deletion lookahead model}.}  

In this paper, we define the \emph{predicted-updates model} to bring \emph{beyond-worst-case analysis} \cite{BWCA-full-book} to the study of dynamic algorithms.  In particular, we provide transformations from 
offline, incremental, 
and decremental dynamic algorithms to fully-dynamic algorithms in our new model.  
In the offline dynamic setting, the entire update sequence is provided to the algorithm at once, and the algorithm's task is to compute all of the solutions corresponding to each day in the sequence.  In the incremental setting, an algorithm only needs to handle element insertions and not deletions.  Similarly, in the decremental setting, an algorithm only needs to handle element deletions.  Our framework can convert an offline dynamic algorithm into a fully-dynamic algorithm using predicted time-stamps for all types of updates.  Our framework can also be extended to convert an incremental algorithm into a fully-dynamic algorithm, with only predictions of deletion events, and a decremental algorithm into a fully-dynamic algorithm, with only predictions of insertion events.  In particular, our result for incremental algorithms generalizes the aforementioned result of \cite{PR23}.

The techniques in our paper bring the study of 
beyond-worst-case analysis for dynamic graph algorithms into focus. 
Much of the study of dynamic algorithms has, thus far, focused on worst-case analysis with many fundamental problems hitting boundaries in efficiency.
Furthermore, state-of-the-art dynamic solutions are often quite complex and require the use of heavy machinery.
The framework we introduce in this paper is fundamentally \emph{simple}, \emph{does not} use heavy machinery, but is applicable
to a \emph{broad} range of problems and settings. 
With the help of predictions, we are able to lift simpler, more implementable, and faster algorithms 
from the offline and partially dynamic settings to the fully dynamic setting. 

Our work also provides motivation for a new 
\emph{noisy setting} in offline and partially dynamic algorithms, with connections to differential privacy~\cite{DMNS06} 
and sensitivity analysis~\cite{VY23,kumabe2022lipschitz}, where 
the input sequence to the algorithm is drawn from a stochastic model; algorithms in these settings may directly have implications in our model.  In particular, our framework gives worst-case guarantees in terms of the \(\ell_1\)-error of the prediction.  This already subsumes an interesting stochastic model, where the realized timestamp of each event is drawn from a distribution centered at the predicted time of the event.  Similar models have been studied for dynamic data structures in \cite{lin2022learning, cao22learning}.

Finally, our framework gives polynomial, 
sometimes even exponential, improvements in runtime over the best fully 
dynamic algorithm for APSP, triconnectivity, dynamic DFS, maxflow/mincut, 
submodular maximization, $k$-edge connectivity,
and others (see~\cref{table:problems}). Our techniques are also inherently parallel (over small $\poly(\log n)$
depth) and may have additional implications in scalable models like the work-depth and distributed models.

\paragraph{Comparison with Concurrent, Independent Work} Recent concurrent work of van den Brand et al.~\cite{BFNP23} 
and Henzinger et al.~\cite{henzinger2023complexity} also study algorithms in dynamic graph models with prediction. 
Henzinger et al.~\cite{henzinger2023complexity} focus on lower bounds in their paper under different prediction 
models from our work. van den Brand et al.~\cite{BFNP23} also study different types of prediction models including
a model very similar to our predicted-deletions model, where instead of using the $\ell_1$-error of a prediction, they 
instead look at the number of element-wise inversions between the predicted \emph{deletion} sequence and the real sequence. 
They give a deterministic algorithm for incremental algorithms with predicted deletions
based on Peng and Rubinstein~\cite{PR23}.  Their reduction maintains the state of an incremental algorithm, in which elements are inserted in approximately reverse deletion order.  In this paper,
we give a reduction from the decremental setting with predicted \emph{insertions} to the incremental setting that can also be 
applied to their result.  The reduction applied to their construction can be interpreted as doing essentially the reverse: it maintains the state of a decremental algorithm, in which elements are deleted in approximately reverse insertion order.

Their framework does not handle offline to fully dynamic transformations for two reasons:
1) it is unclear how to simultaneously insert elements in both approximately reverse deletion order \emph{and} reverse insertion
order when given noisy predictions, and 2) their reduction is tied to incremental algorithms where it is even unclear what types of algorithms
to ask for in the offline setting. We sidestep both issues in our framework with the key data structure used in our 
solution: the random \emph{partition-tree} (\cref{def:random-partition-tree}).  A more detailed comparison of the approaches is provided in \Cref{sec:comparison-to-vdB-etal}

Recent independent and concurrent work of Agarwal and Balkanski \cite{AB23} studies the problem of \emph{dynamic submodular function maximization with cardinality constraints} in a model similar to our predicted update model.  Their update time is given in terms of two parameters: \(w\) which constitutes the magnitude of a ``small" error, and \(\eta\) which is the number of elements with error larger than \(w\).  Their update time is then polylogarithmic in \(\eta\) and \(w\).  This form of update time is incomparable to runtimes in the form of our result.  They also study a different form of submodular maximization than we do in this work, we consider matroid constraints as opposed to cardinality.

\section{Technical Overview and Summary of Contributions}\label{sec:tech}
First, we define and motivate the \emph{predicted-updates dynamic model}.  Then, we introduce and provide a technical overview of our framework to design algorithms for this model.  Finally, we describe specific problems for which our framework is able to outperform state-of-the-art fully dynamic algorithms.  

\subsection{The Predicted-Updates Dynamic Model} 
We introduce the \emph{predicted-updates dynamic model}, which is a general model for algorithms with predictions in the dynamic setting that applies to a wide range of problems in the traditional offline, incremental, and decremental dynamic settings.  
In this paper, we refer to \emph{days} to indicate an ordering to the update operations. The updates can easily be presented as 
an ordered sequence instead of being tied to days. However, we find that speaking in terms of \emph{days} is a useful abstraction 
for describing our results. Throughout, we use \emph{time} and \emph{timestamp} interchangeably with \emph{day}.
We define an \defn{event} to be an update of a certain type on an element $e \in \cS$ where $\cS$ is the ground set of all possible
elements in the dataset. We differentiate between \defn{real} and \defn{predicted} events to be updates that actually 
occur on a day versus an update that is predicted to occur on a day.

\begin{restatable}[Predicted-Updates Dynamic Model]{definition}{predictedupdatesdynamicmodel}
    In the \defn{predicted-updates dynamic model}, we consider a ground set $\cS$ of size $|\cS|$. We are guaranteed updates 
    for three different types of algorithms:

    \begin{enumerate}
        \item Offline Dynamic Algorithms: We are given a \emph{predicted sequence of dynamic updates} $P$ consisting of tuples $(e, type, day, i)$ where
        $e \in \cS$ is the element that the event is performed on, $type$ is the type of event, $day$ is the predicted day, and $i$ (initially 
        set to $0$) is a counter 
        for the number of times the prediction for this event is updated ($i$ is a parameter that is used in our algorithms). The
        predictions can be given to us, online, in sets $P_1, P_2, \dots, P_{\log_2(T)}$ where $P_i \subseteq P_{i + 1}$ and $|P_{i+1}| = 2 \cdot 
        |P_i|$ as we see more real events. In other words, as we see more real events, we receive more predictions. The 
        only requirement on the sets $P_1, \dots, P_{\log_2(T)}$ is that the predicted days cannot be earlier than the day each set of 
        predictions is given.  In particular, it is possible that multiple events are predicted for the same day.  On each day, a real event occurs. \label{item:increasing-bundles}
        \item Incremental Partially Dynamic Algorithms: 
        On each day exactly one of the following occurs:
            \begin{enumerate}
                \item An element \(e \in \cS\) is inserted (a real insertion), 
                and reports a prediction, $(e, day, i)$, of the day on which it will be deleted (\emph{a predicted deletion}).
                \item A previously inserted element is deleted (a real deletion).
            \end{enumerate}
        \item Decremental Partially Dynamic Algorithms: %
        At the outset, the algorithm is given a set $P$ of all the elements that are predicted to ever appear in the system, and predictions for when each of the elements in $P$ will be inserted.
        
        Then, on each day, exactly one of the following occurs:
        \begin{enumerate}
            \item An element \(e\), either in $P$ or not, is inserted, 
            \item A previously inserted element is deleted, and provides a prediction of when it will be reinserted (if ever).  
        \end{enumerate}
        
        As in~\cref{item:increasing-bundles} (with the same restrictions), 
        the original insertion predictions can also
        be given to us in sets $P_1, P_2, \dots, P_{\log_2(T)}$, online, as we see more real deletions. On each day, a real insertion or
        deletion is performed.
    \end{enumerate}
    
An algorithm computes a function \(f(\cdot)\), which is a solution to a problem $\mathcal{P}$, 
in the predicted-updates dynamic model, if on every day \(t\), the algorithm outputs \(f(S)\), where \(S \subseteq \mathcal{S}\) is 
the \emph{working subset} induced by the true (not predicted) sequence of element insertions and deletions that occur in time-steps $t \in [T]$.
\label{def:predicted-updates-dynamic-model}
\end{restatable} 

Note that an algorithm in this model \emph{must} compute each day's output \emph{correctly}.  
The runtime of the algorithm, however, may depend on the prediction error.  
While there is significant recent work surrounding algorithms with predictions, this, along with concurrent, independent works~\cite{BFNP23,henzinger2023complexity}, to the best of our knowledge, are the first to study the dynamic model.\footnote{Specifically, dynamic problems that compute functions over a subset \(S\) of some ground set \(\mathcal{S}\), where \(S\) is subject to insertions and deletions. See e.g.\ this recent survey~\cite{HHS22} for examples of such algorithms.} 
We motivate the predicted-updates dynamic model by examining some key features.

\paragraph{Modularity.}  In many settings of the algorithms-with-predictions domain, the requested predictions are problem-specific.  This is necessary for some problems, where the purpose of the prediction is to provide a partial solution.  However, this requires algorithm designers to provide a custom prediction model, error metric, and algorithm tailored to each new problem.  In our model, our requested predictions of element update times are not problem-specific, and are instead applicable to any dynamic problem.  Hence, techniques can translate across wide classes of problems, as demonstrated by our framework.  We also emphasize that our model makes no distributional assumptions about the input.  Rather, our guarantees are in the form of \emph{worst-case analysis} with respect to a new parameter.  

\paragraph{Practical and efficient machine learning.} A main motivation of algorithms-with-predictions is to take advantage of the predictive power of machine learning models.  Our model and framework allows us to train a model once e.g., to predict insertion/deletion times of edges in a graph
(as a vector of values), and use it many times for a wide range of problems e.g., a collection of dynamic graph problems.  Additionally, our framework can handle predictions that are ill-formed, i.e. not feasible update sequences, so we do not have to take into account possibly complicated feasibility requirements in this induced learning problem. 
 These properties allow us to avoid the costly process of designing and training models that predict problem-specific parameters for each new problem.  Additionally, since our requested predictions do not depend on the solutions to a specific problem \(f(\cdot)\), we \emph{do not} need to compute \(f(\cdot)\) on the dynamic graphs we encounter in the training stage.  This could be a very large saving, as computing \(f(\cdot)\) on each training instance, in many cases, takes time superlinear in the size of the instance.  Finally, we note that a rapidly growing area of empirical research shows that it is possible to use machine learning to generate high-quality predictions of edge insertions and deletions in graphs \cite{kumar2019predicting, nguyen2018continuous,
poursafaei2022towards, quach2021dyglip, rossi2020temporal,wang2021inductive,zhang2018link}.  

\paragraph{Interpolation. }  This model provides a \emph{beyond-worst-case paradigm} for dynamic algorithms that interpolates between the offline and partially dynamic settings (with zero prediction error) and fully dynamic setting (with large prediction error).  
Beyond-worst-case models can give us insight into what bottlenecks are inherent to a problem and model, and inspire new algorithmic techniques that could transfer between models (see \cite{BWCA-full-book} for an in-depth discussion).
We hope that studying the predicted-updates dynamic model can provide insights that inspire better fully dynamic algorithms and lower bounds.  %

\subsection{Framework for Predicted-Updates Dynamic Algorithms}
\label{sec:framework-technical-overview}

In conjunction with our model, we design an algorithmic framework that ``lifts" offline, incremental, and decremental
algorithms to the predicted-updates, fully dynamic setting. In the technical section of our paper, we use the term \emph{work} to 
denote \emph{runtime}.

\begin{theorem}[Predicted-Updates Dynamic Model Framework (Informal)]

    Given a balanced, offline divide-and-conquer algorithm, $\alg$, which performs $\tO(|W|^c)$ work per subproblem $W$ of size $|W|$, 
    we can construct an algorithm in the predicted-updates model that, over $T$ real events, does total work,

    \begin{align*}
        \tO(T + ||\mathrm{error}||_1), \;\;\; \text{when } c \leq 1,\\
        \tO((T + ||\mathrm{error}||_1) \cdot T^{c-1}), \;\;\; \text{when } c > 1,
    \end{align*}

    in expectation and with high probability, where \(||\mathrm{error}||_1\) is the sum over all elements of the absolute difference between the predicted deletion day and the actual deletion day. 
    
    Given an incremental algorithm, $\alg$, with worst-case update time \(\mathrm{update}(\alg)\), we can construct an algorithm in the predicted-updates model that, over \(T\) real events, does total work, in expectation and with high probability,

    \begin{align*}
        \tO\left( (T + ||\mathrm{error}||_1) \cdot \mathrm{update}(\alg) \right).
    \end{align*}

    Given a decremental algorithm, $\alg$, with worst-case update time \(\mathrm{update}(\alg)\), we can construct an algorithm in the predicted-updates model that, over \(T\) real events, does total work, in expectation and with high probability,

    \begin{align*}
        \tO\left( (T + ||\mathrm{error}||_1) \cdot \mathrm{update}(\alg) \right).
    \end{align*}

    Furthermore, given a ``backstop" fully-dynamic algorithm \(B\) for the problem, that does at most \(R_B(T)\) total work by day \(T\), we can get an algorithm that, over \(T\) updates, does total work
    \[\widetilde{O}\left(\min \left\{(T + ||\mathrm{error}||_1) \cdot \mathrm{update}(\alg), \; R_B(T) \right\} \right),\]
    in expectation and with high probability (where $\mathrm{update}(\alg)$ depends on $c$, see above, in the offline case).
    (We use \(\widetilde{O}(\cdot)\) to hide polylogarithmic factors in \(|\cS|\) and \(T\).)
    \label{thm:informal-pred-update-framework}
\end{theorem}

Surprisingly, we are able to achieve all three desiderata of algorithms-with-predictions: consistency, competitiveness, and robustness, simultaneously.  Note also that this guarantee holds \emph{in hindsight}.  That is, the algorithm does not require an estimate of the magnitude of \(||\mathrm{error}||_1\) to achieve this goal.  Additionally, the potential to include a backstop means that, asymptotically, one can design algorithms that can take advantage of predictions \emph{for free}, without losing the assurance of a worst-case guarantee.  We believe that this framework is practically implementable, and does not rely on any algorithmic ``heavy machinery."  This is particularly attractive, as the intention of the model is to provide algorithms that perform well in practice. Furthermore, it is often the case that offline, incremental, and decremental 
algorithms are simpler than their fully dynamic counterparts, allowing us to utilize these algorithms in real-world settings.

\paragraph{The \(\ell_1\) Error Metric}  The results of \Cref{thm:informal-pred-update-framework} give runtime bounds that depend on the \emph{\(\ell_1\) norm} of the prediction error.  We use this as a shorthand to denote the sum of absolute differences between the predicted time and realized time of each event.  One way to interpret this, is that if \(\mathbf{u}\) is a vector indexed by possible events that contains the true update times of each event, and \(\mathbf{p}\) is another vector indexed by possible events that contains the predicted update times of each event, the \(\ell_1\) error of the prediction is 
\(||\mathbf{u} - \mathbf{p}||_1.\)
If an event occurs but is never predicted, or vice versa, that event contributes \(T\) to the \(\ell_1\) error.  Also, for events that occur multiple times, it suffices to consider the closest matching of duplicated events.  

The \(\ell_1\) error metric is a natural one that has been well-studied in applications to warm-starts \cite{DILMV21, DMVW23}.  To contextualize the bound, consider predictions that, on average, are within $\poly(\log(T))$ factor of days of the true event times.  Then, the \(\ell_1\) error will be near-linear in \(T\), and asymptotically, we achieve the runtime of the offline algorithm (up to $\poly(\log(T))$ factors).  On the other hand, the \(\ell_1\) error of any prediction is \(O(T^2)\).  So in the worst case, the algorithm will take runtime equivalent to solving \(T\) independent offline instances (which can be improved by providing the algorithm a worst-case backstop).  

Another related error metric that we could have considered is bounding the runtime by a factor of \(T \cdot ||\mathbf{u} - \mathbf{p}||_\infty\), where \(||\mathbf{u} - \mathbf{p}||_\infty\) is the largest absolute error over all predictions.  This bound is always at least the \(\ell_1\) error and is sometimes comparable in magnitude.  In general, however, it is a significantly weaker bound.  As an example, the \(\ell_1\) error can be (near-)linear in \(T\) even when there are a polylogarithmic number of events, which occur, that were never predicted.  In this same case, \(T \cdot ||\mathbf{u} - \mathbf{p}||_\infty \) would be \(\Omega(T^2)\), giving a very different bound.  Thus, a main focus of our work is getting this tighter dependence on \(\ell_1\) error rather than \(T \cdot \ell_\infty\).

\paragraph{Random Partition-Tree Framework}
Our main technical contribution is the \defn{random partition-tree} data structure (\cref{def:random-partition-tree}). 
Our random partition-tree data structure is a tree drawn uniformly-at-random over the predicted event sequence; such 
a structure crucially mimicks a divide-and-conquer data structure. We show 
that such a simple data structure allows us to \emph{confine the effects of prediction errors ``locally.''} Specifically, we present a way to \emph{fix} the divide-and-conquer data structure on the fly, as prediction errors 
become apparent. We fix the structure by recomputing nodes (representing
recursive subproblems in the divide-and-conquer) in the partition-tree and their descendants. Splitting the recursive subproblems
\emph{randomly} guarantees that the expected amount of recomputation that is triggered by prediction errors is $\tO(|t' - t|)$, where
$t'$ is the predicted time and $t$ is the real time of the event. Hence, we perform as much work per element as the prediction error,
$|t'-t|$, associated with that element. This, fundamentally, allows our total $\ell_1$ prediction error to be as large as the 
number of real updates, $\tO(T)$, while maintaining work matching that of the corresponding offline, incremental, or decremental 
algorithm. %
Our result also makes an interesting connection to \emph{metric tree-embeddings} where our metric of interest is time.
We describe this connection in more detail in~\cref{app:metric-embedding}.

\paragraph{Other Technical Contributions of the Framework} In addition to our main contribution, we also solve a variety of other 
technical challenges listed below. 

\begin{itemize}
    \item \emph{(Early deletion problem):} an element's true deletion time is earlier than its predicted deletion time.  That is, we arrive at day \(t\) to see that element \(e\) is deleted, but we expected \(e\) to be active until some later day \(t' > t\). This means that some subproblems in our divide-and-conquer recursion tree were based on an incorrect list of events, and are therefore not correct.  These subproblems and their descendants need to be recomputed. Our random partition-tree, as described above, allows us to perform the computation
    in work linearly proportional to the prediction error, with small overhead. 

    \item \emph{(Late update problem):} an element's true update time is later than its predicted update time.  That is, we reach a day \(t\) when we expect to see an update to element \(e\), but element \(e\) is not updated on that day (i.e.\ there is no real event on element $e$). 

    We address this by reinserting \(e\) as a prediction for a later time.  We do this via a ``guess-and-double" procedure.  That is, the first time \(e\) is reinserted, we reinsert it for 1 day in the future, then 2 days in the future, then 4 days, etc.  
    Then, we can again fix the data structure to account for this rescheduling using the random partition-tree, like the early deletion problem.

    \item \emph{(Overscheduling problem):} the predicted event times in $P$ could schedule arbitrarily many events for one day \(t \in [T]\). This breaks our divide-and-conquer framework, which distributes work based on the crucial assumption that a small number of events happens per day.  This problem is exacerbated by the repeated rescheduling we do to address the early deletion problem.  

    We show how to reassign the scheduled element deletions, in a preprocessing step, so that not too many are scheduled for one day.   We do this by framing our scheduling problem as an instance of \emph{online metric matching} for the line metric (representing time), and using a known competitive algorithm of \cite{GL12}.  Finally, we account for the rescheduling by assigning a \defn{batch} of \(O(\log{T})\) events to each day, and accounting for these in our total work.

    \item \emph{(Boosting and backstopping for competitiveness):} our work bounds are proven in expectation and our framework may perform worse than state-of-the-art \emph{without backstopping}. We show how to boost our expected work bounds to work bounds with 
    high probability by running $O(\log(T))$ independent instances of our algorithm. 
    Furthermore, we show how to ``backstop" the algorithm by observing that we can compose two dynamic algorithms \(A\) and \(B\) to get one algorithm that, asymptotically, achieves the minimum amortized runtime of both \(A\) and \(B\).  
    We compose our framework with a traditional fully dynamic algorithm using our backstop procedure to get a composite algorithm that has runtime comparable to the state-of-the-art.

    \item \emph{(Handling unknown $T$):} $T$ is often unknown in real-life sequences. We give another ``guess-and-double'' procedure in~\cref{alg:high-probability} that allows us to handle \emph{unknown} $T$, with high probability, by guessing an upper bound on $T$ in 
    successive powers of $2$.

\end{itemize}

The only use of randomness in our framework is for preprocessing our inputs and for constructing our random partition-tree. Implications of using randomness, and potential avenues to derandomize the framework are discussed in~\cref{sec:use-of-randomness}.

\paragraph{Reducing Incremental to Divide-and-Conquer}
Using an incremental algorithm, we show that our framework can be used to design algorithms for a stronger setting, in which elements are inserted fully arbitrarily, and our algorithm only receives predictions about deletion times.  That is, when an element is inserted, it is tagged with a predicted deletion time.  Utilizing our offline framework, we begin by using the incremental algorithm to design an offline dynamic algorithm for this problem.  Our previous theorem shows that we can lift such an algorithm to the setting where we are given predictions for both insertion and deletion events.  We then make two main observations that allow us to drop the need for predictions of insertion events.  
\begin{itemize}
    \item In our original offline to fully-dynamic reduction, we prepare a tentative schedule of events during preprocessing.  This is no longer possible, because we do not have access to any predicted information before the algorithm starts.  However, we achieved this using an \emph{online} algorithm.  Thus, we can use the online algorithm to create the schedule, as we learn about new events. 
    \item In our original offline to fully-dynamic reduction, we run a tentative version of the divide-and-conquer algorithm during preprocessing.  This is no longer possible, again because we do not have access to the predicted events.  We observe that, because of the specific structure of the offline algorithms that arise from incremental algorithms, we can run the computations associated with the nodes of the divide-and-conquer tree using a \emph{just-in-time} approach.  
\end{itemize}

This extension of our framework 
is described in detail in \cref{sec:incremental}.

\paragraph{Reducing Decremental to Divide-and-Conquer}  We also apply our framework to lifting decremental algorithms to the fully-dynamic setting, given predictions of insertion events.  We do this via a simple reduction to the incremental case.  In particular, we reinterpret a decremental algorithm over elements as an incremental algorithm over ``anti-elements."  Then, we are able to directly apply the previous result.  This is described in detail in \cref{sec:decremental}.

\subsection{Applications} %

Our framework transforms a variety of incremental 
algorithms into the predicted-updates, fully dynamic model to obtain guarantees (often)
much better than their fully dynamic counterpart. Our results are 
shown in~\cref{table:problems}.

\begin{table}[htb!]
\renewcommand\arraystretch{2}
    \centering
    \scalebox{0.76}{
    \begin{tabular}{| c || c | c |c|c|}
         \hline
         Problem &  \multicolumn{2}{|c||}{Best Fully Dynamic Runtimes}& \multicolumn{2}{|c|}{New Predicted-Update Runtimes (\cref{thm:predicted-updates-apps,thm:app-runtimes,thm:app-runtimes-deletion})}  \\
         \hline
         \hline
         Planar Digraph APSP & $\tO\left(n^{2/3}\right)$& \cite{fakcharoenphol2006planar, klein2005multiple} & $\tO(\sqrt{n})$ & \cite{das2022near}\\
         \hline
         Triconnectivity & $O(n^{2/3})$ &\cite{GIS99} & $\tO\left(1\right)$ & \cite{holm2020worst,PSS17} \\
         \hline
         $k$-Edge Connectivity & $n^{o(1)}$ & \cite{JS21} & $\tO(1)$ & \cite{chalermsook2021vertex} \\
         \hline
         Dynamic DFS Tree & $\tO\left(\sqrt{mn}\right)$ &\cite{baswana2019dynamic} & $\tO\left(n\right)$ & \cite{baswana2019dynamic,ChenDWZZ18}\\
         \hline
         Minimum Spanning Forest & $\tO(1)$ & \cite{holm2001poly} & $\tO(1)$ & \cite{eppstein1994offline}\\
         \hline
         APSP & \shortstack{$\left(\frac{256}{k^2}\right)^{4/k}$-Approx \\ $\tO\left(n^k\right)$ update\\ $\tO(n^{k/8})$ query} & \cite{FGNS23} & \shortstack{$(2r-1)^k$-Approx \\ $\tO\left(m^{1/(k+1)}n^{k/r}\right)$} & \cite{chen2020fast} \\
         \hline
         AP Maxflow/Mincut & \shortstack{$O(\log(n) \log\log n)$-Approx \\ $\tO\left(n^{2/3 + o(1)}\right)$} & \cite{chen2020fast} & \shortstack{$O\left(\log^{8k}(n)\right)$-Approx.\\ $\tO\left(n^{2/(k + 1)}\right)$} & \cite{Goranci19, GHS19}  \\
         \hline 
         MCF & \shortstack{$(1+\eps)$-Approx \\ $\tO(1)$ update \\ $\tO(n)$ query} & \cite{chen2020fast} & \shortstack{$O(\log^{8k}(n))$-Approx.\\
         $\tO\left(n^{2/(k+1)}\right)$ update \\ $\tO(P^2)$ query} & \cite{Goranci19, GHS19}  \\
         \hline
         Strongly Connected Components & $\Omega(m^{1-\eps})$ \text{query or update} & \cite{abboud2014popular} & $\tO(m)$ & \cite{10.5555/2627817.2627899} \\
         \hline
         Uniform Sparsest Cut & \shortstack{$2^{O(\log^{5/6}(n))}$-Approx \\ $2^{O(\log^{5/6}(n))}$ update \\ $O(\log^{1/6}(n))$ query} & \cite{goranci2021expander} & \shortstack{$O\left(\log^{8k}(n)\right)$-Approx \\ $\tO\left(n^{2/(k+1)}\right)$ \\ $O(1)$ query} & \cite{Goranci19, GHS19} \\
         \hline
         Submodular Max & \shortstack{$1/4$-Approx \\ $\tO(k^2)$} & \cite{dutting2023fully} & \shortstack{$0.3178$-Approx \\ $\tO\left(\poly(k)\right)$} & \cite{FeldmanLNSZ22} \\
         \bottomrule
    \end{tabular}}
    \caption{Table of the best fully dynamic update runtimes for a variety of problems vs.\ our update times obtained via our framework in the 
    predicted-deletion model assuming $||\mathrm{error}||_1 = \tO(T)$. Our query times match the fully dynamic query times for every problem. The acronyms are as follows: APSP: all-pairs shortest
    paths, DFS: depth-first search, AP: all-pairs, MCF: multi-commodity concurrent flows. The variable $P$ is the number of queried pairs in the multi-commodity flow
    result. If only one runtime is shown, the same runtime holds for both update and query. Problem definitions are provided in~\cref{sec:problems}. Here,
    $m$ is the maximum number of edges in the graph at any time. The 
    submodular maximization problem's runtime is measured in terms of the number of function evaluations and $k$ is the rank of the matroid constraint. Some of our results match the state-of-the-art, although, they are potentially simpler
    to implement than the state-of-the-art so we present them for completeness.} %
    \label{table:problems}
\end{table}

\section{Offline Divide-and-Conquer to Fully-Dynamic Transformation}

\subsection{Preliminaries}

In this paper, we use the following notation throughout. We denote the \emph{total} number of updates in our update sequence by $T$.
We use $[t]$ to denote the set of positive integers up to $t$, ie.\ $\{1, \dots, t\}$. We use $\tO(\cdot)$ to hide $\poly(\log(|\cS| \cdot T))$ factors where $\cS$ is our ground set for the elements and $T$ is the total number of updates. Often, it holds that $|\cS| = \poly(T)$.
Our explicit bounds are presented in the formal theorem statements.

Below, we first define the \defn{partition-tree} (which is generated randomly in
our algorithm), the main data structure in our framework. 

\begin{definition}[Partition-Tree]\label{def:random-partition-tree}
A partition-tree of a sequence of days is a binary tree such that every node in the tree is associated with a window (interval of time), with the following properties:
\begin{enumerate}
    \item The root window of the tree is the full sequence of \(T\) days.  
    \item For each internal (parent) window \(W\), its children windows partition \(W\). We refer to the parent of window \(W\) as \(\mathrm{parent}(W)\). 
    \item Each leaf window is associated with a single day.
\end{enumerate}
We will refer to the size of a window associated with a node in the partition-tree as the \emph{size} of the node. Henceforth, we will refer to all nodes in 
the partition-tree as \defn{windows}.
\end{definition}

Our framework works for any \defn{$c$-divide-and-conquer} algorithm defined below. For our offline to fully dynamic transformations, we 
show a number of interesting algorithms fall under this definition. For our incremental and decremental to fully dynamic transformations, 
we show that \emph{any} incremental or decremental algorithm with \emph{worst-case} update time can be framed as a $c$-divide-and-conquer 
algorithm and can, hence, be used in our framework. 

\begin{definition}[Divide-and-conquer algorithm]
Given a partition-tree \(T\), a \(c\)-\emph{divide-and-conquer} algorithm \(\alg\), for some constant \(c > 0\),
computes a function \(f(\cdot)\) (to solve a problem $\mathcal{P}$) 
over subsets of the ground set \(\mathcal{S}\) (generally, $|\cS| = \poly(T)$), 
such that: %
\begin{enumerate}
    \item A polynomial amount of read-write memory \(\mathbf{M}\) is given to the root window; 
    \item Algorithm \(\alg\) associates a computation with each window \(W\) of $T$, where the computation is given only the \emph{unordered} set of events occurring in \(W\), and read-write access to \(\mathbf{M}_{\mathrm{parent}(W)}\) which is the state of \(\mathbf{M}\) after \(\mathrm{parent}(W)\) has terminated; Notably, \(\mathbf{M}\) is passed  down to \(W\) on a path from the root, and is \emph{not} mutated by any sibling windows;
    
    \item And the \emph{expected} work done at window \(W\) is \(O(\Gamma \cdot |W|^c)\), where $\Gamma$ represent
    $O(\poly\log(|\cS| \cdot T))$ factors or $\text{update}(\alg)$ where $\alg$ is a worst-case incremental or decremental algorithm.\footnote{$\text{update}(\alg)$ is the worst-case update time of an incremental or decremental algorithm.}
    For simplicity of expression in our proofs,\footnote{We omit the factor since $\Gamma$ shows 
    up repeatedly in all of our expressions, and this expression would fall under our $\tO(\cdot)$ bound
    for the offline transformations and is explicitly given for our incremental/decremental transformations.} 
    we sometimes omit $\Gamma$ in our calculations.\label{item:balance-condition}  %
    \item The solution to $\mathcal{P}$ is obtained from the windows in the tree and the memory, $\mathbf{M}$, within each window.
\end{enumerate}
\label{def:divide-and-conquer}
\end{definition}

The requirement in \Cref{item:balance-condition} ensures that the work done by the divide-and-conquer algorithm is evenly distributed across the tree.  This is the analogue of our condition for the incremental/decremental to fully-dynamic setting, where we require the incremental algorithm to have a worst-case (not amortized) bound.

\begin{algorithm}
\caption{Fully Dynamic Algorithms with Predictions from Offline Divide-and-Conquer Algorithms}\label{alg:predicted-insertion-deletion}
    \begin{algorithmic}[1]
        \Require{\emph{Offline (during preprocessing):} partition-tree $\mT$, 
        divide-and-conquer algorithm $\alg$ that computes $f(\cdot)$, and predicted sequence of dynamic updates $P$. 
        \emph{Online:} Online sequence of dynamic updates $U = [E_1, \dots, E_T]$ where each event $E = (e, type)$ is a \emph{real} event.}
        \Ensure{After each day $t \in [T]$, output $f(U_t)$.}

        \State $\matchP \leftarrow \textsc{PreprocessPredictions}(P)$. \Comment{Converts 
        $P$ into a feasible set of predictions}\label{main:preprocess}
        \State Insert all events in $\matchP$ as \emph{predicted} events on their corresponding days.
\\
    \Procedure{retrigger}{$t_1, t_2$}
        \State Find $W$, the smallest window in $\mT$ that contains both $t_1$ and $t_2$.\label{offline:find-W}
        \State $S \leftarrow \{W\}$ \Comment{Set of windows to process.}
        \While{$S \neq \emptyset$}
            \State Remove a window $W'$ from $S$.
            \For{each child $C$ of $W'$}
                \State $\textsc{ProcessEvents}(C)$.\label{offline:process-event} %
                \State $S \leftarrow S \cup \{C\}$.
            \EndFor
        \EndWhile
    \EndProcedure
\\
    \Procedure{ProcessEventEarlierThanPrediction}{$E, t, t_{\text{predict}}$}
        \State Remove \emph{predicted} event $E$ from day $t_{\predict}$.\label{offline:delete-predict}
        \State Add $E$ to day $t$ as a \emph{real} event. \label{offline:add-real-event}
        \State \textsc{Retrigger}$(t, t_{\text{predict}})$.\label{offline:earlier-retrigger}
    \EndProcedure
  \\  
    \Procedure{ProcessEventLaterThanPrediction}{$E, i, t$}
        \State Remove \emph{predicted} event $E$ from day $t$.\label{offline:remove-early-prediction}
        \State Add $E$ to day $(t+2^i)$ as a \emph{predicted} event.\label{offline:reschedule-later-prediction}
        \State Change the prediction for $E = (e, T)$ in $\matchP$ 
        to $(e, T, t+2^i, i)$.\label{offline:change-later}
        \If{$e$ has a corresponding \emph{predicted} deletion event that occurs earlier than $(t + 2^i)$}\label{offline:out-of-order}
            \State Remove $e$'s \emph{predicted} deletion event from the day $t'$ that it was scheduled for.\label{offline:reschedule-1}
            \State Add $e$'s \emph{predicted} deletion event to day $(t+2^i)$.\label{offline:reschedule-2}
            \State Let $(e, deletion, t', i_{del})$ be the entry for $e$'s deletion event in $P$.\label{offline:reschedule-3}
            \State Change $e$'s deletion event prediction in $\matchP$ to $(e, deletion, (t + 2^i), i_{del} + 1)$.\label{offline:reschedule-4}
        \EndIf
        \State \textsc{Retrigger}$(t, t+2^i)$.\label{offline:later-retrigger}
    \EndProcedure
\\
    \For{day $t \in [T]$}\label{offline:iterate-day}
        \For{each \emph{real} event $E = (e, T)$ on day $t$}\label{offline:real}
            \State Suppose event $E$ is an event of type $T$ on element $e$.\label{offline:type-element}
            \State Find the corresponding prediction $(e, T, t_{\predict}, i) \in \matchP$ for event $E$.\label{offline:find-prediction}
            \If{$t < t_{\predict}$}\label{offline:t-less-than-predict}
                \State \textsc{ProcessEventEarlierThanPrediction}$(E, t, t_{\predict})$.\label{offline:call-process-earlier}
            \EndIf
        \EndFor
        \For{each \emph{predicted} event $E = (e, T)$ on day $t$}\label{offline:predict}
            \If{event $E$ wasn't a \emph{real} event on day $t$}\label{offline:not-real}
                \State \textsc{ProcessEventLaterThanPrediction}$(E, i+1, t)$.
            \EndIf
        \EndFor
        \State Return \textsc{Output}($t$).\label{offline:output}
    \EndFor
    \end{algorithmic}
\end{algorithm}

\subsection{Algorithm Description}\label{sec:offline-alg-description}

We obtain two separate categories of inputs to the algorithm: offline inputs received during preprocessing and online inputs.
\cref{alg:predicted-insertion-deletion} uses the partition-tree $\mT$, the divide-and-conquer algorithm $\alg$ that computes $f(\cdot)$, 
and the predicted sequence of dynamic updates $P$. The predicted sequence is given as tuples $(e, type, day, i)$ where $e$ is the element,
$type$ is the type of update (an insertion or deletion), $day$ is the predicted day of the update, and $i$ is the number of 
times the event has been rescheduled. Initially, all $i = 0$ for all events in the prediction sequence. 
Then, the 
actual set of updates is given as an online sequence of dynamic updates $U$. The algorithm outputs an answer to the function $f(\cdot)$
after each update on each day $t \in [T]$ where $U_t$ consists of the set of updates that occurred on days $[1, \dots, t] \subseteq [T]$.

Given a raw set of update predictions $P$, we must first convert it
into a feasible sequence of events in our model with error comparable
to the original set of update predictions. In particular, the 
error of our converted sequence of predictions has error at most 
$O(\log(T))$ factor worse than the original set $P$. 
In~\cref{main:preprocess}, we do this preprocessing. We include
the pseudocode for $\textsc{PreprocessPredictions}(P)$ in~\cref{appendix}.

To do this, we first convert $P$ into an instance of metric online bipartite 
matching where $P$ represents the requests and the set of days $t \in 
[T]$ are our servers. We would like to produce a matching such that 
at most one event in $P$ is matched to each $t \in [T]$. We use the harmonic algorithm of Gupta and Lewi~\cite{GL12}, which 
iteratively matches each request to its nearest open server on the 
right or left, with harmonic probabilities.  We can run the iterations 
efficiently using a union-find data structure to produce a matching via the
online bipartite matching procedure. 
 Then, 
finally, we do a post-processing on the produced matching. In
the order of the days, we perform a linear scan to check if any deletions
of events occur before their respective insertions. For any such 
deletions, we move the event to the same day as the insertion. We 
can do this in $O(T)$ time by maintaining hash maps keyed by the 
events. After doing this post-processing, we obtain our feasible 
sequence of events $\matchP$ which consists of at most 
two events per day: at most one insertion and at most one deletion.
Then, we insert each predicted event in $\matchP$ into their corresponding
predicted days. This procedure can also be applied to any constant number of event types 
with precedence constraints (i.e.\ one event type must be performed before another).

\cref{alg:predicted-insertion-deletion} iterates through each day $t \in [T]$ (\cref{offline:iterate-day}) and checks for each \emph{real} (\cref{offline:real})
and \emph{predicted} event (\cref{offline:predict}) assigned to day $t$. A \emph{real}
event is an event that is assigned to day $t$ by the update sequence. A 
\emph{predicted} event is one that assigned by our algorithm (i.e.\ an event
that is predicted to occur on day $t$). For each real event of 
type $T$ and on element $e$ (\cref{offline:type-element}), we find the 
corresponding prediction for this event in $\matchP$ (\cref{offline:find-prediction}). 
If this event is not in $\matchP$, then we assign the default prediction of 
$(e, T, \infty, 0)$ where $\infty$ indicates that we predict the event to occur
at the very end of the update stream. If our current day $t$ is less than 
the predicted day $t_{\predict}$ (\cref{offline:t-less-than-predict}), then
we process the event as one that occurs \emph{earlier} than predicted. 
We use the procedure $\textsc{ProcessEventEarlierThanPrediction}(E, t, 
t_{\predict})$ (\cref{offline:call-process-earlier}), explained below. 

After iterating through the real events, we now iterate through 
the \emph{predicted} events on day $t$ (\cref{offline:predict}).
If the predicted event $E = (e, T)$ was not a real event that occurred on 
day $t$ (\cref{offline:not-real}), then, we update our prediction since the 
real event will occur on a later day. To update our prediction, we call
the procedure, \textsc{ProcessEventLaterThanPrediction},
for handling events that occur \emph{later} than the predicted
day. Finally, we call the function $\textsc{Output}$ that is obtained from 
$\alg$ for computing the output of our dynamic algorithm at $t$ 
(\cref{offline:output}). 

We now describe each of our individual procedures that is called within our main
algorithm:

\begin{itemize}
    \item \textbf{$\textsc{Retrigger}(t_1, t_2)$}: This procedure retriggers 
    the processing of all events in \emph{every descendant} of the smallest window 
    $W$ in $\mT$ that contains both $t_1$ and $t_2$. The processing of the events uses the procedure \textsc{ProcessEvents}
    (\cref{offline:process-event})
    that is obtained from $\alg$.
    \item \textbf{ProcessEventEarlierThanPrediction$(E, t, t_{\text{predict}})$}:
    This procedure processes a real event that occurs on a day $t$ prior to its
    predicted day $t_{\text{predict}}$. First, we delete the predicted event $E$
    from day $t_{predict}$ (\cref{offline:delete-predict}). Then, we 
    add $E$ to the day $t$ as a real event (\cref{offline:add-real-event}).
    Finally, we find the smallest window in $\mT$ that contains both $t$ and $t_{\text{predict}}$ (\cref{offline:find-W}); let this window be $W$.
    We then call retrigger on $W$ (\cref{offline:earlier-retrigger}).
    \item $\textbf{ProcessEventLaterThanPrediction}(E, i, t)$: 
    This procedure reschedules events which occur later than the prediction; specifically, the event is rescheduled
    to a future day.%
    We first remove the predicted event $E$ from day $t$ (\cref{offline:remove-early-prediction}). We then increase the predicted day for the corresponding event 
    from $t$ to $t+2^i$ where $i$ the number of times it has been rescheduled
    before (\cref{offline:reschedule-later-prediction}). Then, we find the smallest window $W$
    that contains both $t$ and $(t + 2^i)$ (\cref{offline:find-W})
    and change the corresponding prediction to reflect the new prediction (\cref{offline:change-later}). We now need to manage settings where there are
    multiple possible updates associated with an element and ordering constraints among
    the types of updates on the elements. In the case of edge insertions and deletions, 
    the insertion for an edge must happen before the deletion of the edge (such is an 
    example of an ordering constraint). Thus, we look for all corresponding predicted deletion
    events that occurs on an earlier day than $(t + 2^i)$ (\cref{offline:out-of-order}) and
    reschedule them to day $(t + 2^i)$ (\cref{offline:reschedule-1,offline:reschedule-2,offline:reschedule-3,offline:reschedule-4}). Note that the conditional statement
    given in~\cref{offline:out-of-order} can apply for any ordering constraints on 
    any type of updates, not only insertions and deletions. Finally, retrigger on $W$ 
    (\cref{offline:later-retrigger}). 
\end{itemize}

\subsection{Preprocessing and Scheduling of Predictions}
Our preprocessing step in~\cref{main:preprocess} of~\cref{alg:predicted-insertion-deletion} performs two main functions.  The first is to create an initial schedule of events.  The second is to compute an initial set of solutions for these scheduled events using a divide-and-conquer algorithm.  Both the schedule and the solutions will be updated throughout the run of the algorithm.  

We begin by rescheduling the input predictions to ensure feasibility.  That is, the input predictions could be malformed--either by having many events predicted for a single day, or by having deletion events predicted to be before insertion events. Having too many predicted events on a single day leads to issues in showing our expected work bounds using 
our partition-tree analysis (discussed later). Thus, we preprocess these predictions to form a feasible sequence of events, that will be a well-formed input to our divide-and-conquer algorithm.  We show that it is possible to do this in near-linear time, while only blowing up the \(\ell_1\) error of the predictions by a polylogarithmic factor.

The ability to handle ill-formed predictions is a key functionality of our algorithm.  This allows the associated learning problem to be the simple problem of learning a vector that is close to our target prediction vector in \(\ell_1\) distance, without having to account for potentially complicated feasibility constraints. 

We can formulate our predicted sequence of dynamic updates to a prediction vector $\predvec$ by assigning each event two unique positions in $\predvec$, one for insertions
and one for deletions. Then, in the corresponding position, we predicted day in $[T]$ that the event occurs. Then, our preprocessing produces a vector $\assignedvec$ where
each day contains at most two events (at most one insertion and one deletion). We can also create a vector $\realvec$ of the real events that occur online. 
For each real event that was not predicted in $P$, we assign the event positions at the end of the vectors $\assignedvec$ and $\realvec$ and give them all a predicted
day of $\max(\realvec) + 1$, which means that we predict the event to occur at the very end, after the last real event. Note that we generate these vectors \emph{solely} for the
purpose of computing the $\ell_1$ error. Our preprocessing procedure only needs to output a sequence $\matchP$ consisting of tuples where predictions in $P$ are assigned to days
such that at most two events occur on any day.

\begin{restatable}[Initial Scheduling Quality and Runtime]{lem}{initial}
\label{lem:preprocessing-schedule}
Let $\predvec$ be the vector of predictions, mapping each event to a day in \([T]\).  Let $\realvec$ be the (unknown to the algorithm) true vector of real events, mapping each event to the day in \([T]\) that it actually occurs.

Given $\predvec$, we can compute an assignment vector $\assignedvec$ such that $\assignedvec$ assigns at most one insertion event and at most one deletion event to each day, 
assigns all insertions events before deletion events, and has error
\[\mathbf{E}\left[ ||\assignedvec - \realvec||_1 \right] = O \left(||\predvec - \realvec||_1 \cdot \log (T)\right). \]

Furthermore, we can compute this in time \(O(T \log^* (T))\).  
\end{restatable}

The formal proof of this lemma is given in~\cref{app:preprocess}.  At a high level, the problem of reallocating the events so that there is at most one on each day is an instance of bipartite matching over the line metric.  We wish to solve this instance in time near-linear in \(T\), otherwise this preprocessing could dominate the total runtime of the algorithm.  We observe that, interestingly, the online algorithm of \cite{GL12} can be used as a fast \(O(\log T)\)-approximation algorithm for this offline problem.  

For predictions where deletions are listed before their corresponding insertions, we perform one additional pass of $\predvec$ and reassign all deletion events which occur before their 
corresponding insertion event to the day the corresponding insertion event occurs. This produces an assignment where at most two events are assigned to the same day and requires 
only $O(T)$ additional work. 

After we obtain a feasible sequence of predictions, $\matchP$, we now construct our partition-tree over our $c$-divide-and-conquer algorithm $\alg$ (\cref{def:divide-and-conquer}). 
We now compute the work of constructing our partition-tree, $\mT$, using $\matchP$ and prove the depth of the constructed tree. We first prove the work to 
compute the partition-tree.

\begin{lemma}[Work to Compute Initial Divide-and-Conquer Solutions]
The full divide-and-conquer algorithm over the partition-tree can be computed in expected time 
\begin{align*}
    O( T^c \cdot \log^{3} (T) \cdot \log\log(T \cdot |\mathcal{S}|)) \;\;\; &\text{when } c > 1 \\
    O(T \cdot \log^3 (T) \cdot \log\log(T \cdot |\mathcal{S}|)) \;\;\; &\text{when } c \leq 1.
\end{align*}
\label{lem:preprocessing-partial-solutions-work}
\end{lemma}

\begin{proof}
    This is a special case of \Cref{lem:single-retrigger-work} where we can think of the original windows that are computed in $\mT$ from $\matchP$ are computed as a result of calling \(\textsc{Retrigger}(1, T)\). 
\end{proof}

We now bound the depth of the partition-tree over the predictions in $P$. Bounding the depth is necessary in order to minimize the work of performing recomputation when the 
real update sequence deviates from the predictions. As the below analysis draws inspiration from the analysis of randomized quicksort, we relegate its proof to~\cref{app:preprocess}.

\begin{restatable}[Partition-Tree has \(O(\log(T))\) Depth]{lemma}{partitiontreedepth}
\label{lem:partition-tree-depth}
    The depth of a random partition-tree drawn over \(T\) days is \(O(\log (T))\) in expectation. Furthermore, for any constant \(k \ge 36\),
    the depth is
    \[\le k \ln(T) \qquad \qquad \text{with probability} \qquad \ge 1 - T^{-\frac{k}{24}}.\]
\end{restatable}

Note that we do not need knowledge of $T$ when computing the aforementioned bounds. We show in~\cref{alg:high-probability} how to 
remove the assumption on the value of $T$.

Finally, we give the work of maintaining the schedule after predictions are rescheduled based on the online dynamic sequence 
of updates. Predictions are rescheduled using \textsc{ProcessEventEarlierThanPrediction} and \textsc{ProcessEventLaterThanPrediction}. Below (in~\cref{lem:rescheduling} with proof deferred to~\cref{app:preprocess}) 
we show the maximum total number of times all predicted events 
are rescheduled by~\cref{alg:predicted-insertion-deletion}
which combined with the work necessary to call \textsc{Retrigger}
which we analyze in~\cref{sec:one-retrigger,sec:total-retrigger} gives the total work of our algorithm.

\begin{restatable}[Work to Maintain Schedule]{lemma}{maintainschedule}\label{lem:rescheduling}
\cref{alg:predicted-insertion-deletion} performs at most 
\(O(T \log T)\) reschedules to the predictions of 
updates and maintain the predicted schedule of events over the course of the online dynamic updates.  
\end{restatable}

\subsection{Work of a Single Call to \textsc{Retrigger}}\label{sec:one-retrigger}
The main workhorse of our algorithm is the \(\textsc{Retrigger}(t_1, t_2)\) operation, 
which recomputes the subtree of the partition-tree rooted at the smallest window $W$ 
that contains $t_1$ and $t_2$. 
This operation recomputes the part of the partition-tree that 
can be affected by rearranging events that are scheduled between \(t_1\) and \(t_2\). 
We showed above that because the partition-tree is drawn randomly, we can bound the depth of the tree
with high probability. Using what we proved above, in this section, we show the additional properties that:

\begin{itemize}
    \item The expected work to recompute the subtree associated with a \textsc{Retrigger}$(t_1, t_2)$ call scales linearly with \(|t_2 - t_1|\),
    \item The work of the divide-and-conquer algorithm is balanced over calls to \textsc{Retrigger}, and
    \item Hence, combined with the fact that our tree has bounded $O(\log T)$ depth, with high probability,
    our running time scales \emph{linearly} with the $\ell_1$ error.
\end{itemize}

We first show that the expected work to recompute the subtree associated with a \textsc{Retrigger}$(t_1, t_2)$ call scales with $|t_2 - t_1|$. To do this, we 
first prove the expected size of the smallest window that contains any pair of days $t_1$ and $t_2$. Such a lemma allows us to bound the amount of computation we 
need to perform to recompute the events between any two $t_1, t_2 \in [T]$. 

\begin{restatable}[Random Partition-Tree Preserves Lengths in Expectation]{lemma}{randomtreepreserveslengths}
\label{lem:random-tree-preserves-lengths}
    Consider a random partition-tree drawn according to \Cref{def:random-partition-tree} for time sequence \([T]\).  Then, for any \(t_1, t_2 \in [T]\) where $t_1 \neq t_2$, the expected size of the smallest window that strictly contains both \(t_1\) and \(t_2\) is \(O(|t_2 - t_1| \cdot \log(T))\).
\end{restatable}

\begin{proof}
We prove this lemma via a coupling argument where we show another way to generate the same distribution of binary tree partitions that will be easier to analyze.  
For each possible divider \(d \in [T - 1]\) of the original sequence, we associate 
\(d\) with a rank $r_d$ drawn uniformly at random from the uniform distribution over $[0, 1]$. The ranks are drawn independently for each $d$.

Using the ranks, we assign the tree structure from the top down.  At the top level, the divider with the lowest rank is used to split the sequence into the left child and the right child.  Iteratively, for each new window of the tree, we use the lowest ranked divider of its subsequence to split the sequence.  
This results in the same distribution over partition-trees as \Cref{def:random-partition-tree}, as at each level, each of the possible dividers is equally likely to be chosen next.  

Now, we see that a contiguous subsequence \([t_\text{start}, t_\text{end}]\) of \([T]\) is a window in the partition-tree if and only if the divider directly preceding \(t_\text{start}\) and the divider directly following \(t_\text{end}\) both have lower ranks than all of the dividers between \(t_\text{start}\) and \(t_\text{end}\),
where we can consider the endpoints of the original sequence over all days in $[T]$ 
to be dividers with rank $0$. This is illustrated in \Cref{fig:tree-splitting-by-ranking}.

\begin{figure}
    \centering
    \includegraphics[scale=0.39]{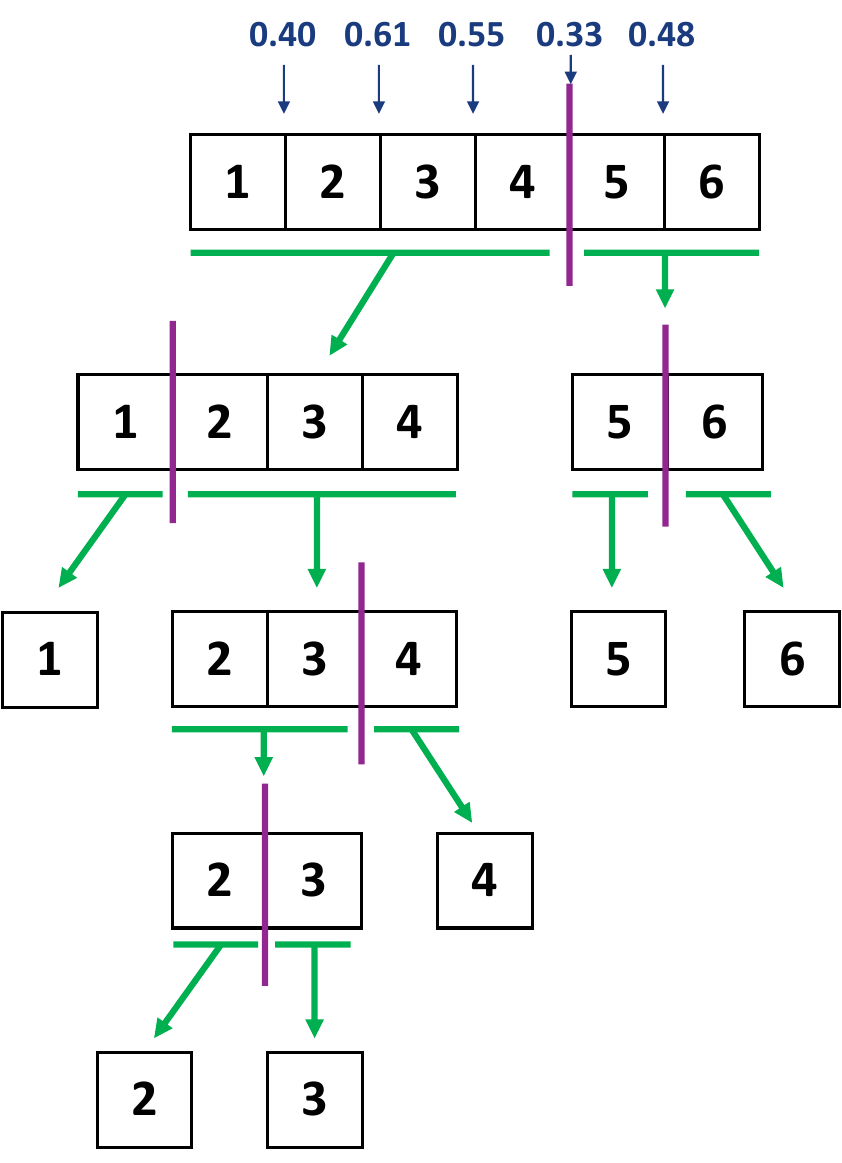}
    \caption{Dividers are drawn uniformly at random from $[0, 1]$ and shown in blue above the figure. This 
    example contains $5$ dividers since $T = 6$.
    \([2, 4]\) is a window in the tree because \(0.40, 0.33 < 0.61, 0.55\).  \([2, 5]\) is not a window in the tree because \(0.48 \not< 0.33\).}
    \label{fig:tree-splitting-by-ranking}
\end{figure}

Without loss of generality let \(t_1 < t_2\).  Let \(W\) be the smallest window that strictly contains \(t_1\) and \(t_2\).  We have that the size of 
$W$ (in days) is
\[|W| = (t_2 - t_1 + 1) + L + R,\]
where \(t_2 - t_1 + 1\) are the number days between the dividers bordering $t_1$ and $t_2$. \(L\) is a random variable representing the number of days before \(t_1\) until we reach one that is bordering a divider with strictly smaller rank than the \(t_2 - t_1 + 2\) dividers drawn between \(t_1\) and \(t_2\) and the dividers bordering $t_1$ 
and $t_2$. 
Symmetrically, \(R\) is the random variable for days after \(t_2\) until we reach a divider with smaller rank. 
This is illustrated in \Cref{fig:L-and-R-definitions}. $L$ and $R$
must each contain dividers with ranks \emph{larger} than the dividers between $t_1$ and $t_2$; otherwise, $W$ would not be the smallest
window strictly containing $t_1$, $t_2$, and all days in between.

\begin{figure}
    \centering
    \includegraphics[scale=0.5]{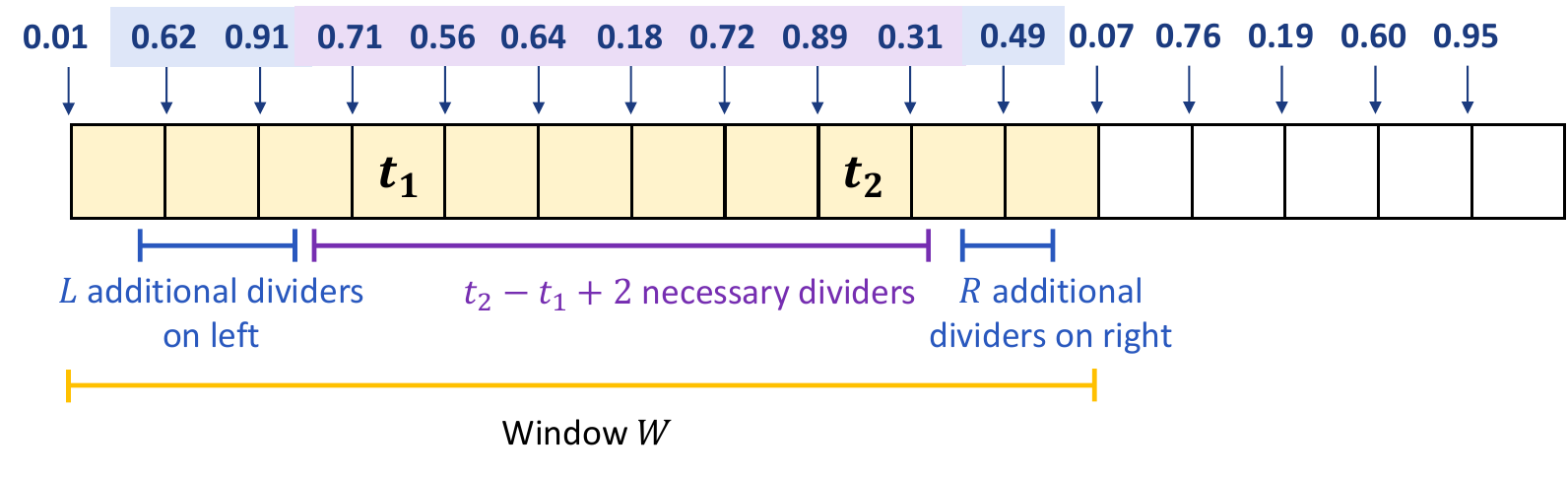}
    \caption{$W$ is the smallest window containing $[t_1, t_2]$. 
    The random variables \(L\) and \(R\) capture the number of additional days added to \([t_1, t_2]\) to satisfy the property illustrated in~\cref{fig:tree-splitting-by-ranking}. The dividers bordering $W$ must have smaller value than all dividers in $W$. Furthermore, the dividers in $L$ and $R$ have \emph{larger} values
    than the dividers bordering $[t_1, t_2]$ and the dividers contained within $[t_1, t_2]$; otherwise, $W$ would not be the smallest window
    containing $[t_1, t_2]$.}
    \label{fig:L-and-R-definitions}
\end{figure}

First, we analyze \(\mathbf{E}[L]\).  Let \[Z = \min \{r_d : d \in [t_1 - 1, t_2 + 1]\}\] be the minimum rank of the dividers between \(t_1\) to \(t_2\) and 
including the ones bordering $t_1$ and $t_2$.  We can compute $\expect[L]$ by conditioning on different values of \(Z\).  Let \(p_Z(\cdot)\) be the p.d.f. of \(Z\). 
\begin{align}
    \mathbf{E}[L] &= \int_{z = 0}^1 \mathbf{E}[L | Z = z] \cdot p_Z(z) dz \label{eq:expect-1}\\
    &\le \int_{z = 0}^1 \min \left\{ T, \frac{1}{z} \right\} p_Z(z) dz \label{eq:expect-2}\\
    &\le \int_{z = 0}^{1/T} T dz + \int_{z = 1/T}^1 \frac{1}{z} \cdot p_Z(z) dz\label{eq:expect-3} \\
    &\le 1 + \int_{z = 1/T}^1 \frac{1}{z} \cdot p_Z(z) dz\label{eq:expect-4}\\
\end{align}
\cref{eq:expect-2} follows because at most \(T\) items can be added to the window before we run out of items in the sequence to add; furthermore, 
the $\frac{1}{z}$ term is the mean of the geometric distribution with parameter $z$ since the probability that we draw a rank smaller 
than $z$ is at most $z$. We can simplify our equation in~\cref{eq:expect-3} to~\cref{eq:expect-4} since $\int_{z = 0}^{1/T} T dz = 1$.

Now, we can find the p.d.f.\ of \(Z\) via its c.d.f.  Let \(F_Z(\cdot)\) be the c.d.f.\ of \(Z\). 
\begin{align*}
     F_Z(z) &= \textbf{P}[Z < z] = 1 - (1 - z)^{t_2 - t_1 + 2}\\
     p_Z(z) &= \frac{d}{dz} \left(F_Z(z)\right) \\
     &= (t_2 - t_1 + 2) (1 - z)^{t_2 - t_1 + 1} . 
\end{align*}

This allows us to conclude the following since $\int_{z = 1/T}^1 \frac{1}{z} dz = \ln(T)$:
 \begin{align*}
     \mathbf{E}[L] &\le 1 + \int_{z = 1/T}^1 \frac{(t_2 - t_1 + 2)(1 - z)^{t_2 - t_1 + 1}}{z} dz \\
     &\le 1 + (t_2 - t_1 + 2) \int_{z = 1/T}^1 \frac{1}{z} dz &\text{since }(1 - z)^{t_2 - t_1 + 1} \le 1\\
     &= 1 + (t_2 - t_1 + 2) \ln(T).
\end{align*}

Symmetrically, by the same argument, \(\mathbf{E}[R] \le 1 + \ln(T)\) as well.  So by linearity of expectation, we have
\begin{align*}
    \mathbf{E}[W] &= (t_2 - t_1 + 1) + \mathbf{E}[L] + \mathbf{E}[R] \\
    &\le (t_2 - t_1 + 1) + 2\cdot (t_2 - t_1 + 2) \cdot \ln(T) \\
    &= O(|t_2 - t_1| \cdot \log(T)).
\end{align*}
\end{proof}

The proof of this lemma is morally similar to a randomized tree-embedding scheme, in the vein of \cite{FRT03}.  We discuss this connection in a later section.

Now, we are ready to bound the work of the \(\textsc{Retrigger}\) operation described in~\cref{sec:offline-alg-description}. 
The work is defined in terms of a $c$-divide-and-conquer algorithm as defined in~\cref{def:divide-and-conquer}.

\begin{lemma}[Expected Work of \(\textsc{Retrigger}\) Operation]\label{lem:single-retrigger-work}
Provided a $c$-divide-and-conquer algorithm (\cref{def:divide-and-conquer}), 
the expected work of \(\textsc{Retrigger}(t_1, t_2)\) is bounded by 
\begin{align*}
    \tO\left(|t_2 - t_1| \cdot T^{c - 1}\right) \;\;\; &\text{for } c > 1 \\
    \tO\left(|t_2 - t_1|\right) \;\;\; &\text{for } c \leq 1,
\end{align*}

where $W_{t_1, t_2}$ is the smallest window containing $t_1$ and $t_2$ and all days in between.
\end{lemma}

\begin{proof}
    Consider a call to \(\textsc{Retrigger}(t_1, t_2)\). We first show the following observation that our algorithm 
    given in~\cref{alg:predicted-insertion-deletion} assigns at most $O(\log(T))$ (real or predicted) events to each day.
    First, by~\cref{lem:preprocessing-schedule}, after preprocessing, we assign at most two predicted events to each day. 
    Then, predicted events become reassigned by either \textsc{ProcessEventEarlierThanPrediction} or \textsc{ProcessEventLaterThanPrediction}.
    An event gets processed by \textsc{ProcessEventEarlierThanPrediction} when a real event happens on a day earlier than the prediction. 
    By the guarantee that at most one real event occur on any day, at most one predicted event can be moved to a day $t \in [T]$ by
    \textsc{ProcessEventEarlierThanPrediction}. Then, \textsc{ProcessEventLaterThanPrediction} reassigns a predicted event to a 
    later day when an event occurs later than predicted. By our procedure for reassigning such events, we keep track of a counter $i$ for 
    each predicted event, denoting the number of times it has been reassigned. Because $\assignedvec'$ contains at most $2$ events per
    day and \textsc{ProcessEventLaterThanPrediction} reassigns events in increments of powers of $2$, each day $t \in [T]$ gets assigned
    at most two events at a distance of $2^i$ away for all $i \in [\log(T)]$. Hence, the number of events assigned to any day is $O(\log(T))$.
    We call this set of events, on day $t$, a \defn{batch} of events and denote it as $\batch_t$.

    First, we fix the size of the smallest window containing $[t_1, t_2]$ and let \(S = |W_{t_1, t_2}|\).
    Now, we consider the work done in any one level $\ell$ of the subtree rooted at $W_{t_1, t_2}$. 
    Let the windows in level $\ell$ of the subtree have sizes \(S_1, \dots, S_k\).  If \(c > 1\), the work done at this level can be bounded by 
    \[\sum_{i= 1}^k S_i^c \le \sum_{i = 1}^k S_i \cdot T^{c - 1} = S T^{c - 1} ,\] since $S_i \leq T$.
    If \(c \leq 1\), the work of the subtree rooted at $W_{t_1, t_2}$ is dominated by the root, and so, the work done at level $\ell$ is bounded
    by
    \[\sum_{i = 1}^k S_i  \le S .\]
    
    By~\cref{lem:partition-tree-depth}, if we set \(k = 24(c + 2)\), we get that the depth of the random partition-tree is \(24 (c + 2) \cdot O(\log(T))\) over $T$ days 
    with probability \(\ge 1 - T^{-(c + 2)}\)  (since \(c\) is a fixed constant).\footnote{We can choose to set $k$ to be an arbitrarily
    large constant if we want to increase the probability of success.} 
    Conditioning on this event, and accounting for the \(\log(\log(T \cdot |\mathcal{S}|))\) blowup from using persistent 
    memory, we get that the total work is bounded by 
    \begin{align} 
        O\left(S \cdot T^{c - 1} \log (T) \log(\log(T \cdot |\mathcal{S}|))\right) \;\;\; & \text{when } c > 1 \label{conditioned-work-1}\\
        O\left(S \log(T) \log(\log(T \cdot |\mathcal{S}|))\right) \;\;\; & \text{when } c \leq 1,\label{conditioned-work-2}
    \end{align}

    where the extra $\log(T)$ comes from multiplying our work per level by $O(\log(T))$ levels.

    Each day corresponds to a batch of \(O(\log T)\) events, so there are at most \(|t_2 - t_1| \log T\) events between \(t_1\) and \(t_2\). By~\cref{lem:random-tree-preserves-lengths}, we can bound the expected size of $S = |W_{t_1, t_2}|$ to 
    be $\expect[S] = O(|t_2 - t_1| \log^2(T))$. Thus, if we condition on the event that the depth of the partition-tree is $O(\log(T))$, then
    the expected work, using~\cref{conditioned-work-1,conditioned-work-2}, denoted by $\text{work}(W_{t_1, t_2})$, is given by

    \begin{align}
        \expect[\text{work}(W_{t_1, t_2}) &\mid \text{ partition-tree has depth } O(\log T)] =\\
        &\begin{cases}
            O(|t_2 - t_1| \cdot T^{c - 1} \log^3(T) \log(\log(T \cdot |\mathcal{S}|))) \;\;\; & \text{when } c > 1 \\
            O(|t_2 - t_1| \cdot \log^3(T) \log(\log(T \cdot |\mathcal{S}|))) \;\;\; & \text{when } c \leq 1.
        \end{cases}\label{work-cases}
    \end{align}

    In the case where the tree is not \(O(\log T)\) depth, which happens with probability at most \(T^{-(c + 2)}\), we have the trivial bound that the tree can be depth at most \(T\).  Thus, the work over the subtree can be bounded as 

    \begin{align*}
        \expect[\text{work}(W_{t_1, t_2}) &\mid \text{ partition-tree \emph{does not} have depth } O(\log T)] =\\
        &\begin{cases}
            O(|t_2 - t_1| \cdot T^{c} \log^{3}(T) \log(\log(T \cdot |\mathcal{S}|))) \;\;\; & \text{when } c > 1 \\
            O(|t_2 - t_1| \cdot T \log^3(T) \log(\log(T \cdot |\mathcal{S}|))) \;\;\; & \text{when } c \leq 1.
        \end{cases}
    \end{align*}
    
    We can also trivially bound that \(|t_2 - t_1| \le T\), allowing us to trivially bound the expected work by 

    \begin{align}
        \expect[\text{work}(W_{t_1, t_2}) &\mid \text{ partition-tree \emph{does not} have depth } O(\log T)] =\\
        &\begin{cases}
            O(T^{c + 1} \log^{3}(T) \log(\log(T \cdot |\mathcal{S}|))) \;\;\; & \text{when } c > 1 \\
            O(T^2 \log^3(T) \log(\log(T \cdot |\mathcal{S}|))) \;\;\; & \text{when } c \leq 1.
        \end{cases}\label{work-cases-unbounded-depth}
    \end{align}

    Using the two cases given in~\cref{work-cases,work-cases-unbounded-depth} allows us to bound the expected work over this subtree by 
    \begin{align*}
        \expect[\text{work}(W_{t_1, t_2})] &= \expect[\text{work}(W_{t_1, t_2}) \mid \text{depth } O(\log T)] \cdot \prob[\text{depth } O(\log T)] \\
        &+ \expect[\text{work}(W_{t_1, t_2}) \mid \text{\emph{not} depth } O(\log T)] \cdot \prob[\text{\emph{not} depth } O(\log T)]\\
        & = (1 - T^{-(c + 2)}) \cdot O((|t_2 - t_1| \cdot \log(T)) \cdot T^{c - 1} \log^{3}(T) \log(\log(T \cdot |\mathcal{S}|)))\\ 
        &+ T^{-(c + 2)} \cdot O(T^{c + 1} \log^{3}(T) \log(\log(T \cdot |\mathcal{S}|))) \\
        &\qquad = O(|t_2 - t_1| \cdot T^{c - 1} \log^{3}(T) \log(\log(T \cdot |\mathcal{S}|))) & \text{when } c > 1 \\ \\
        & = (1 - T^{-(c + 2)}) \cdot O((|t_2 - t_1| \cdot \log(T)) \cdot \log^3(T) \log(\log(T \cdot |\mathcal{S}|)))\\
        &+ T^{-(c + 2)} \cdot O(T^2 \log^3(T) \log(\log(T \cdot |\mathcal{S}|))) \\
        &\qquad = O(|t_2 - t_1| \cdot \log^3(T) \log(\log(T \cdot |\mathcal{S}|))) &\text{when } c \leq 1.
    \end{align*}

    This proves our desired bounds.
\end{proof}

\subsection{Total Work Over All Calls to \textsc{Retrigger}}\label{sec:total-retrigger}

In the previous subsection, we bounded the work done by a single call to \(\textsc{Retrigger}\).  In this section, we bound the total work done by calls to \(\textsc{Retrigger}\) over the entire run of~\cref{alg:predicted-insertion-deletion}.  

\begin{lemma}[Expected Work Over All Calls to \(\textsc{Retrigger}\)]\label{lem:total-retrigger}
The expected work done by~\cref{alg:predicted-insertion-deletion} over all calls to \textsc{Retrigger} is 
\begin{align*} 
    \tO(|\mathbf{p} - \mathbf{r}|_1 \cdot T^{c - 1}) \;\;\; & \text{when } c \ge 1 \\ 
    \tO(|\mathbf{p} - \mathbf{r}|_1) \;\;\; & \text{when } c < 1. 
\end{align*}
\end{lemma}

\begin{proof}
    First, consider a single event \(e\).  Let \(\mathbf{a}_0 (e)\) be the day that \(e\) is assigned to after the preprocessing step (\Cref{lem:preprocessing-schedule}).  Let \(\mathbf{r}(e)\) be the true day on which \(e\) occurs.  

    If \(\mathbf{r}(e) \le \mathbf{a}_0 (e)\), then \(\textsc{Retrigger}(\mathbf{r}(e), \mathbf{a}_0 (e))\) is 
    called exactly once for \(e\).  Otherwise, there is a sequence of reassignments, following a doubling search procedure, over \(\mathbf{a}_1(e), 
    \dots, \mathbf{a}_k(e)\), where \(\mathbf{a}_{i+1}(e) = 2^i + \mathbf{a}_i (e)\), \(\mathbf{a}_i(e) \le \mathbf{r}(e)\) for \(i < k\), and the 
    final \(\mathbf{a}_k(e) \ge \mathbf{r}(e)\).
    \(\textsc{Retrigger}\) is called for each of these reassignments: \(\textsc{Retrigger}(\mathbf{a}_i(e), \mathbf{a}_{i + 1}(e))\) for \(i = 0, \dots, k - 1\), and finally called one last time as \(\textsc{Retrigger}(\mathbf{r}(e), \mathbf{a}_k(e))\).  
    Thus, by linearity of expectation, the total expected work done is 
    \begin{align}
        \left[\sum_{i = 0}^{k - 1} \mathrm{work}(\textsc{Retrigger}(\mathbf{a}_i(e), \mathbf{a}_{i + 1}(e)))\right] + \mathrm{work}(\textsc{Retrigger}(\mathbf{r}(e), \mathbf{a}_k (e))).
        \label{eq:retrigger-raw-work-1}
    \end{align}

    Assume for simplicity we are in the \(c > 1\) case.  (The \(c \leq 1\) analysis simply doesn't contain the $T^{c - 1}$ term.) We use~\cref{lem:single-retrigger-work} to bound \Cref{eq:retrigger-raw-work-1} by
    \begin{align*}
        &\left[\sum_{i = 0}^{k - 1} O\left(|\mathbf{a}_{i + 1}(e) - \mathbf{a}_i (e)| \cdot T^{c - 1} \log^{3}(T) \log\log(T \cdot |\cS|)\right)\right] + O\left(|\mathbf{a}_k(e) - \mathbf{r}(e)| \cdot T^{c - 1} \log^{3}(T) \log\log(T \cdot |\cS|)) \right)\\
        &\leq \left( \sum_{i = 0}^{k - 1} |\mathbf{a}_{i + 1}(e) - \mathbf{a}_i (e)| + |\mathbf{a}_k(e) - \mathbf{r}(e)| \right) \cdot O \left( T^{c - 1} \log^{3} (T) \log\log(T \cdot |\mathcal{S}|)\right) \\
        &= O\left( |\mathbf{a}_0(e) - \mathbf{r}(e)| \cdot T^{c - 1} \log^{3} T \log\log(T \cdot |\mathcal{S}|) \right). 
    \end{align*}
    
    The final inequality follows by first 
    observing that 
    \begin{equation}
        |\mathbf{r}(e) - \mathbf{a}_{k}(e)| \le |\mathbf{r}(e) - \mathbf{a}_{0}(e)|, 
        \label{eq:error-nonincreasing-with-level}
    \end{equation} 
    by the nature of the doubling search.  This is because \(\mathbf{a}_k(e)\) only exists if \(\mathbf{a}_{k - 1}(e) < \mathbf{r}(e)\), and \(\mathbf{a}_{k - 1}(e) = \frac{1}{2}(\mathbf{a}_0(e) + \mathbf{a}_{k}(e))\).
    
    This means that 
    \begin{align*}
        \sum_{i = 0}^{k - 1} |\mathbf{a}_{i + 1}(e) - \mathbf{a}_i(e)| &= |\mathbf{a}_k(e) - \mathbf{a}_0(e)| \\
        &\le |\mathbf{a}_k(e) - \mathbf{r}(e)| + |\mathbf{r}(e) - \mathbf{a}_0(e)|  \\
        &\le 2|\mathbf{a}_0(e) - \mathbf{r}(e)|  &\text{by (\ref{eq:error-nonincreasing-with-level})}
    \end{align*}
    So in total, 
    \[\sum_{i = 0}^{k - 1} |\mathbf{a}_{i + 1}(e) - \mathbf{a}_i (e)| + |\mathbf{a}_k(e) - \mathbf{r}(e)| \le 3 |\mathbf{a}_0(e) - \mathbf{r}(e)| = O(|\mathbf{a}_0(e) - \mathbf{r}(e)|).\]

    This bounds the expected cost of calls to \(\textsc{Retrigger}\) related to the specific event \(e\).
    To the bound the work over all events \(e\), recall we denote the vector of input predictions as \(\mathbf{p}\). Then we can bound the total work by 
    \begin{align*}
        &\sum_{e} O\left( |\mathbf{a}_0(e) - \mathbf{r}(e)| \cdot T^{c - 1} \log^3(T) \log\log(T \cdot |\mathcal{S}|) \right) \\
        &= |\mathbf{a}_0 - \mathbf{r}|_1 \cdot O\left(T^{c - 1} \log^{3}(T) \log\log(T \cdot |\mathcal{S}|) \right) \\
        &= O \left(|\mathbf{p} - \mathbf{r}|_1 \cdot T^{c - 1} \log^{4}(T) \log\log(T \cdot |\mathcal{S}|)\right) &\text{by \cref{lem:preprocessing-schedule} scheduling quality, when } c > 1.
    \end{align*}

    Using the same analysis and~\cref{lem:single-retrigger-work}, 
    we can bound the total work by 
    \[O \left(|\mathbf{p} - \mathbf{r}|_1 \cdot \log^4(T) \log\log(T \cdot |\mathcal{S}|)\right)), \qquad \qquad \text{when } c < 1.\]
\end{proof}

\subsection{Boosting to High Probability}

\subsubsection{Best-of-All-Worlds Backstop}
\label{sec:backstop}

In this section we describe how to compose $N$ dynamic algorithms $\{A_1, \dots, A_N\}$, that achieve amortized guarantees leading to 
a framework that achieves both \emph{competitiveness} and our desired work bound \emph{with high probability} (described in~\cref{sec:work-bounds-high-probability}).
That is, by day \(t\), algorithm \(A_i\) has done total work \(R_{A_i}(t)\).  
We show how to use $\{A_1, \dots, A_N\}$ to design one algorithm that by day \(t\) has done total work 
\[O\left(N \cdot \min \{A_1, \dots, A_N\} \right).\]

The merit of implementing such a backstop depends on the use-case of the algorithm.  In some cases, we would hope that a predicted-updates algorithm can significantly outperform the fully-dynamic algorithm in some but not all regimes, in which case the backstop may be useful.  In other cases, the motivation could be simplicity of implementation.  The framework presented in this work does not rely on algorithmic ``heavy machinery," and in some cases, could be more practical to implement without a backstop.

Below, we define \emph{one computation step} as one word operation in the word-RAM.  The reduction is straightforward, and relies heavily on the fact that the goal is an amortized bound, and not a worst-case update time: we run the algorithms ``in parallel," and we return the output of the algorithm that terminates first.  We show our desired work bounds in~\cref{thm:backstop} 
and give the formal proof in~\cref{app:backstop}.

\begin{algorithm}
\caption{Backstop Meta-Algorithm for $N$ Algorithms}\label{alg:backstop}
\begin{algorithmic}[1]
\State Initialize algorithms $\{A_1, \dots, A_N\}$.
\State Initialize event buffers for algorithms $\{A_1, \dots, A_N\}$.
\State %

\For{day \(t\), event \(e_t\) occurs} \Comment{occurs online}
    \State Add \(e_t\) to event buffers for $\{A_1, \dots, A_N\}$. \label{algline:add-event-to-buffer}
    \State Iteratively perform one computation step of each algorithm in $\{A_1, \dots, A_N\}$ 
    until one of the algorithms has completed the computation for all events in its buffer. \label{algline:interleave-computation}
    \State Output the completed computation.
\EndFor
\end{algorithmic}
\end{algorithm}

\begin{restatable}[Best-of-All-Worlds Backstop]{thm}{backstop}\label{thm:backstop}
    Fix $N$ algorithms $\{A_1, \dots, A_N\}$, that will see the same, a priori unknown, input sequence over \(T\) days.
    Define \(R_{A_i}(t)\) to be the total work that \(A_i\) does over the first \(t\) days. %

    We can design a meta-algorithm \(M\), such that 
    \[\forall t, \quad R_M(t) = O \left(N \cdot \min \{R_{A_1}(t), \dots, R_{A_N}(t)\} \right),\]
    where \(R_M(t)\) is the total work that \(M\) does over the first \(t\) days.   
\end{restatable}

\subsubsection{Obtaining Work Bounds with High Probability}\label{sec:work-bounds-high-probability}

In this section, we show how to modify our partition-tree-based algorithm using our backstop meta-algorithm
to obtain our running time bounds with high probability, using a simple ``boosting" argument.  We first make the observation that if we run $O\left(\log \left(T\right)\right)$ independent instantiations of our algorithm, at least one instantiation will have runtime close to the expectation with high probability. 

We can also remove the assumption that \(T\) is known, by using an additional guess-and-double argument; such an argument comes in handy for our incremental or decremental
transformation to fully dynamic in later sections. 

\begin{lemma}[Boosting Argument]\label{lem:boosting}
    Consider \(k \log (T)\) independent instantiations of the predicted-updates algorithm (\cref{alg:predicted-insertion-deletion}), for some constant integer \(k > 0\).  With probability at least \(1 - \frac{1}{T^{k}}\), at least one of these instantiations has runtime 
    \begin{align*}
        \tO\left(T^{c-1} \cdot \left(T + ||\mathbf{p} - \mathbf{d}||_1\right) \right) \;\;\; &\text{ when } c > 1, \\
        \tO\left(T + ||\mathbf{p} - \mathbf{d}||_1\right) \;\;\; &\text{ when } c \leq 1.
    \end{align*}
\end{lemma}

\begin{proof}
    Consider a single instantiation of the predicted-updates algorithm given in~\cref{alg:predicted-insertion-deletion}.  
    By~\cref{lem:preprocessing-partial-solutions-work,lem:total-retrigger}, the total expected work of this algorithm (including 
    preprocessing) is 
    \begin{align}
        O\left(T^{c-1} \cdot \left(T + ||\mathbf{p} - \mathbf{d}||_1\right) \cdot \log^{c+3} (T) \cdot \log\log(T \cdot |\mathcal{S}|)\right) \;\;\; &\text{ when } c > 1, \label{eq:expected-runtime-1} \\
        O\left( \left(T + ||\mathbf{p} - \mathbf{d}||_1\right) \cdot \log^{4}(T) \cdot \log\log(T \cdot |\mathcal{S}|) \right) \;\;\; &\text{ when } c \leq 1. \label{eq:expected-runtime-2}
    \end{align}
    By Markov's inequality, the probability that the runtime of a single instantiation exceeds two times the expected work is at most \(\frac{1}{2}\). Thus, the probability that $k \log(T)$ independent instantiations all simultaneously exceed two times the expected work is 
    \[ \le \left(\frac{1}{2}\right)^{k \log (T)} = \frac{1}{T^{k}}.\]
    Thus, with probability at least \(1 - \frac{1}{T^{k}}\), one of the \(k \log (T)\) instantiations has runtime at two times the expected work given in~\cref{eq:expected-runtime-1,eq:expected-runtime-2}.
\end{proof}

Now, using the backstop technique, we take \(O(\log(T))\) instantiations of the predicted updates algorithm (\cref{alg:predicted-insertion-deletion}) to create a composite algorithm that has runtime that scales with the minimum of the instantiations.  

We observe that the guarantee from~\cref{lem:preprocessing-schedule} and~\cref{lem:total-retrigger}, and therefore~\cref{lem:boosting} as well, requires the algorithm to have access to the size of the time horizon \(T\).  We show how to set this via a guess-and-double procedure
assuming we receive sets of updates $P_1, P_2, \dots, P_{T}$ of successively larger sizes where 
$P_i \subseteq P_{i+1}$ and $|P_{i+1}| = 2 \cdot |P_i|$.  Together, this gives us~\cref{alg:high-probability} and our final theorem below.

\begin{algorithm}
\caption{Work Bounds with High Probability}\label{alg:high-probability}
\begin{algorithmic}[1]
\Require{\emph{Offline (during preprocessing):} partition-tree $\mT$, 
    divide-and-conquer algorithm $\alg$ that computes $f(\cdot)$, predicted sequence of dynamic updates $P$, and ground set $\mathcal{S}$
    where each event is on an element $e \in \mathcal{S}$. 
    \emph{Online:} Online sequence of dynamic updates $U = [E_1, \dots, E_T]$ where each event $E = (e, type)$ is a \emph{real} event; prediction sets $P_1, \dots, P_{\log_2(T)}$.}
\Ensure{After each day $t \in [T]$, output $f(U_t)$.}
\State \(\hatT \leftarrow 1\) \Comment{Initialization}
\State \(L \leftarrow k\cdot \log(\hatT)\) for some fixed constant $k > 0$.
\For{real event $E_t$}
    \If{$t \geq \hatT$} \Comment{Guess-and-double} \label{algline:guess-and-double}
        \State Obtain prediction set $P_{\max(1, \log_2(\hatT))}$.
        \State $\hatT \leftarrow 2 \cdot \hatT$.
        \State \(L \leftarrow \max(k \cdot \log(\hatT), \log(|\mathcal{S}|))\) for sufficiently large constant $k > 0$.\footnotemark
        \State Create \(\mathcal{A} := [A_1, \dots, A_L]\), \(L\) independent instantiations of~\cref{alg:predicted-insertion-deletion} using $\hatT$ and $P_{\max(1, \log_2(\hatT))}$. \Comment{Instantiations use independent sources of randomness.}
        \State $M \leftarrow$ Initialize backstop meta-algorithm (\cref{alg:backstop}) over $\alg$. 
        \State Use \(M\) to process all previously seen events $E_1, \dots, E_{t-1}$ in \(B\).
    \EndIf
    \State Add $E_t$ to \(B\).
    \State Pass $E_t$ to \(M\), return output of \(M\).
\EndFor
\end{algorithmic}
\end{algorithm}
\footnotetext{The $\log(|\mathcal{S}|)$ term is necessary to ensure we have enough copies to ensure with high probability.}

\subsection{Putting Everything Together: the Final Theorem}

We use~\cref{alg:high-probability} and all our prior proofs to show our final reduction. We use a $c$-divide-and-conquer algorithm
as well as a partition-tree which are given in~\cref{def:divide-and-conquer} and~\cref{def:random-partition-tree}, respectively.
The online predictions are given in sets $P_1, P_2, \dots, P_{\log_2(T)}$ where $P_1 \subseteq P_2$ and $|P_2| = 2 \cdot |P_1|$. 
Prediction set $P_i$ is given (before the next real update) when we have processed $|P_{i-1}|$ real updates. Let $U_{t}$ denote
the set of events in $[E_1, \dots, E_t]$.

\begin{theorem}[Offline Divide-and-Conquer to Fully Dynamic with Predictions Reduction]
\label{thm:offline-to-fully-dynamic}

Given a $c$-divide-and-conquer algorithm $\alg$ that computes the solution to a problem $\mathcal{P}$, 
sets of predictions $P_1, P_2, \dots, P_{\log_2(T)}$, and
online sequence of events $U = [E_1, \dots, E_t]$,~\cref{alg:high-probability} 
(using \cref{alg:predicted-insertion-deletion}) correctly outputs the 
solution to $f(U_t)$ after each $t \in [T]$ and uses total work, with high probability, 

\begin{align*}
    \tO\left(T^{c-1} \cdot \left(T + ||\mathbf{p} - \mathbf{d}||_1\right) \right) \;\;\; &\text{ when } c > 1, \\
    \tO\left(T + ||\mathbf{p} - \mathbf{d}||_1\right) \;\;\; &\text{ when } c \leq 1.
\end{align*}

Furthermore, our algorithm never performs worse than the best-known fully dynamic algorithm for problem $\mathcal{P}$. 

\end{theorem}

\begin{proof}
Given our inputs, we construct our set of random partition-trees using~\cref{alg:predicted-insertion-deletion}.
We first prove that our algorithm correctly outputs $f(U_t)$ after every event $E_t$ in the sequence of online events.
~\cref{alg:high-probability} runs $L$ instantiations of~\cref{alg:predicted-insertion-deletion} and outputs the answer 
from the first instantiation that finishes processing event $E_t$. By our procedures in~\cref{alg:predicted-insertion-deletion},
every real update is processed in every window that contains it, and no windows in the union of all windows which do not contain 
days after $t$ contain any predicted events that have \emph{not} occurred after event $E_t$ is processed. Furthermore, by~\cref{def:divide-and-conquer}, all windows process all events irrespective of the order of the events. Hence,~\cref{alg:predicted-insertion-deletion}
correctly returns $f(U_t)$ for every $t \in [T]$.

We now give our high probability bound for the total amount of work performed over all $T$ real events. 
    Consider a day \(t = 2^i\) on which \Cref{algline:guess-and-double} is triggered, and \(\widehat{T}\) is doubled.  \(\widehat{T}\) is newly set to be \(2^{i + 1}\).  Consider the total work done by the new instantiation of $M$. Specifically, the work between the event at \(t = 2^i\) and when \(\widehat{T}\) is doubled again at event \(t' = 2^{i + 1}\).  By~\cref{thm:backstop}, we can bound the work of \(M\) by \(L\) times the work of the minimum \(A_i \in \mathcal{A}\) that is drawn independently 
    in this iteration.  \Cref{lem:boosting} bounds the work of that minimum \(A_i\) by

    \begin{align*}
        \tO\left(T^{c-1} \cdot \left(T + ||\mathbf{p} - \mathbf{d}||_1\right) \right) \;\;\; &\text{ when } c > 1, \\
        \tO\left(T + ||\mathbf{p} - \mathbf{d}||_1\right) \;\;\; &\text{ when } c \leq 1,
    \end{align*}
    
    with probability $\geq \max(1 - \frac{1}{t^k}$ if $\log(t) > \log(|\mathcal{S}|)$ and $\geq 1 - \frac{1}{t^{k \cdot \log(|\mathcal{S}|)}}$,
    otherwise. The probability is due to the fact that we create $L = \max(k\cdot \log(T), k \cdot \log(|\mathcal{S}|))$ independent instantiations.

    Now, can use a union bound to conclude that our work bounds hold for all integers $t \in [T]$ with probability 
    \begin{align*}
        &\ge 1 - \left( \sum_{i = 0}^{\infty} \frac{1}{2^{k(i + 1)}}  + \frac{\log(T)}{2^{k \cdot \log{|\mathcal{S}|}}}\right) \\
        &\ge 1 - \frac{2\log(T)}{|\mathcal{S}|^k},
    \end{align*}

    which is high probability for sufficiently large constant $k > 0$. 

    This condition allows us to, for a day \(t\), bound the total work done by \Cref{alg:backstop} up through day \(t\).  
    Define \(i\) such that \(2^i \le t < 2^{i + 1}\).  We can bound the work done by \Cref{alg:backstop} up through day \(t\) as
    
    \begin{align*}
        \sum_{i = 0}^{\floor{\log_2(t)}} O\left(|\mathbf{p} - \mathbf{r}|_1 \cdot (2^{i})^{c - 1} \log^{c+3}(2^{i}) \log \log(2^{i})\right)
        =O\left(|\mathbf{p} - \mathbf{r}|_1 \cdot (t)^{c - 1} \log^{c+3}(t) \log \log(t \cdot |\mathcal{S}|)\right), 
    \end{align*}

    when $c > 1$,

    and $O\left(|\mathbf{p} - \mathbf{r}|_1 \cdot \log^{4}(t) \log \log(t)\right)$ when $c \leq 1$. Substituting $T$ for $t$ into the equations
    results in our desired work.

    Finally, by composing \Cref{alg:high-probability} with the best-known algorithm for $\mathcal{P}$, once again using \cref{thm:backstop}, 
    we never perform worse than the best-known algorithm for the problem.
\end{proof}

Finally, we observe that this framework gives a stronger in the update/query model, when we reduce to an algorithm that has worst-case query time.  That is, an algorithm that handles updates \emph{without} information about queries, can continue to do so in this model.  This is in contrast to some offline divide-and-conquer algorithms that use information about what elements will be queried to make decisions.

\begin{corollary}[Bound for Update/Query problems]\label{cor:update-query-work}
    Given a $c$-divide-and-conquer algorithm $\alg$ for \(f(\cdot)\) with worst-case query time $\mathrm{query}(A)$, 
    we can construct an algorithm, such that with high probability with respect to $T$, has total work 
    
    \begin{align*}
        \tO\left(T^{c-1} \cdot \left(T + ||\mathbf{p} - \mathbf{d}||_1\right) \right) \;\;\; &\text{ when } c > 1, \\
        \tO\left(T + ||\mathbf{p} - \mathbf{d}||_1\right) \;\;\; &\text{ when } c \leq 1,
    \end{align*}
    and the worst-case query time is always bounded by $\mathrm{query}(A)$.

    Furthermore, our algorithm never performs worse than the best-known fully dynamic algorithm for problem $\mathcal{P}$.
\end{corollary}

\begin{proof}
    For the first part of the corollary, we can think of an update/query problem as simply being an update problem, where after each update, the algorithm must return a data structure that is compatible with the query algorithm.  Using~\cref{thm:offline-to-fully-dynamic} on this update problem, we get a predicted-dynamic algorithm that always returns a data structure that is compatible with the original query algorithm.  Thus, the query algorithm is unchanged. 

    For the second part of the corollary, we use the backstop procedure of \Cref{thm:backstop} to backstop the predicted-deletion algorithm with the fully-dynamic algorithm.  This provides some data structure for every update.  However, our offline to online query algorithm and the fully-dynamic query algorithm, $\mathrm{query}(B)$, may not be expecting the same data structure.  Thus, to execute a query on a given day, we run, in parallel, both query algorithms. 
    Thus our query time is bounded by $O\left( \min \{\mathrm{query}(A), \mathrm{query}(B) \}\right)$.
\end{proof}

\section{Incremental to Fully-Dynamic Transformation}\label{sec:incremental}

In this section we show how our framework that lifts offline algorithms to the fully-dynamic setting, given predictions of all update times, can be adapted to lift an incremental algorithm to the fully-dynamic setting, given predictions of only the deletion times. 
Formally, we consider the following model.  

\begin{definition}[Predicted-Deletion Dynamic Model]
    In the predicted-deletion dynamic model, we consider a ground set \(\mathcal{S}\).  On each day exactly one of the following occurs:
\begin{enumerate}
    \item An element \(e \in \mathcal{S}\) is inserted, and reports a \emph{prediction} of the day on which it will be deleted.
    \item A previously inserted element is deleted.
\end{enumerate}
An algorithm computes a function \(f(\cdot)\) in the predicted-deletion dynamic model, if on every day \(t\), the algorithm outputs \(f(S)\), where \(S \subseteq \mathcal{S}\) is 
the \emph{working subset} induced by the true (not predicted) sequence of element insertions and deletions that occur in time-steps $1, \dots, t$.
\label{def:predicted-deletion-dynamic-model}
\end{definition}

This model is stronger than the general predicted-updates model that was presented in the last section.  In the previous section, we charged the runtime of the fully-dynamic algorithm to the prediction error in insertions and deletions.  In this model, the runtime of the fully-dynamic algorithm should only depend on the predictions of deletion times.  This corresponds to the fact that we are starting with an incremental algorithm rather than an offline algorithm.  Thus, the algorithm we start with can already handle arbitrary insertions, and our reduction adds functionality only in handling deletions.  

First, in what may seem like a step in the wrong direction, we show how to build an \emph{offline dynamic} divide-and-conquer algorithm, using an incremental algorithm with a worst-case guarantee.  
To do this, we introduce the concept of permanent elements.  

\begin{definition}[Permanent elements]
    An element is \emph{permanent} with respect to a window \(W\) of a partition-tree if
    \begin{enumerate}
        \item the element is inserted on or before the first day of \(W\),
        \item the item is deleted after the last day of \(W\), 
        \item and the element is \emph{not} permanent for any ancestor of \(W\).
    \end{enumerate}
\end{definition}

\begin{figure}
    \centering
    \includegraphics[scale=0.39]{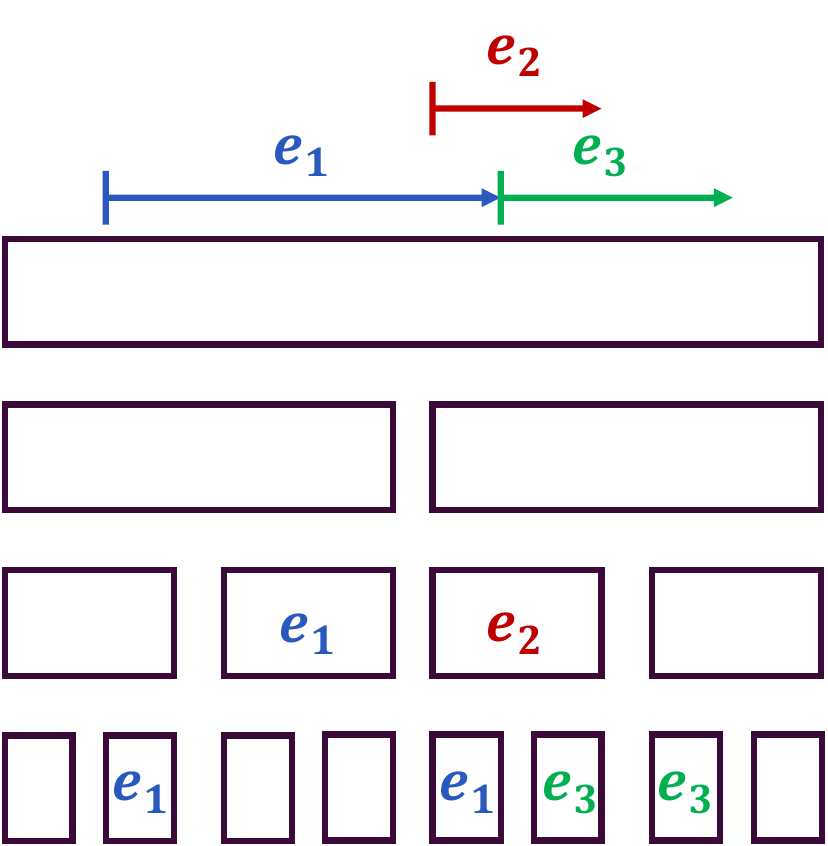}
    \caption{An illustration of which windows consider the elements \(e_1\), \(e_2\), and 
    \(e_3\) permanent}
    \label{fig:permanent-elements}
\end{figure}

This is illustrated in \Cref{fig:permanent-elements}.  Note that for a day \(t\), the windows containing \(t\) perfectly partition the elements that are active on day \(t\).  Conversely, we also have that the windows for which an element \(e\) is permanent perfectly partition the lifetime of \(e\).  This suggests a natural way to convert an (online) incremental algorithm, $\alg$, with worst-case update time \(\mathrm{update}(\alg)\) into a divide-and-conquer offline algorithm.

\begin{lemma}[Offline-dynamic divide-and-conquer from incremental]
    An incremental algorithm, $\alg$, with worst-case update time \(\mathrm{update}(\alg)\), can be converted into a offline-dynamic divide-and-conquer algorithm satisfying \Cref{def:divide-and-conquer}, where the work at each window, \(W\), is 
    \[O(\mathrm{update}(\alg) \cdot |W|).\]
    \label{lem:divide-and-conquer-from-incremental}
\end{lemma}

\begin{proof}
  At the root node, we initialize the state of the incremental algorithm.  Then, at each window, we insert the permanent elements associated with that window into the state of the incremental algorithm.  The set of permanent elements can be computed from the set of all elements that are active at any time during the parent window, in a linear-time scan.  Thus, the work done at a window scales with \(\mathrm{update}(\alg)\) times the number of permanent edges associated with that window.
  
  Note that any element that is permanent for a window \(W\), must have either been inserted or deleted during \(W\)'s sibling window.  Otherwise, this element would actually be permanent for \(W\)'s parent window.  Thus, the number of permanent elements of \(W\) is bounded by the size of \(W\)s sibling.  For a random partition-tree, this is the same as the size of \(W\) in expectation.  
  So this meets the third criterion of \Cref{def:divide-and-conquer}, where the expected work at a node is linear in the size of the node, with a multiplicative factor of \(\mathrm{update}(\alg)\).  

  With respect to the second criterion of \Cref{def:divide-and-conquer}, we need that the computation of the window depends only on the state of the parent window, and the \emph{unordered set} of updates occurring in this window.  At first glance, this does not appear to be the case, as calculating the set of permanent elements requires the algorithm to compare the set of updates of its sibling window.  However, we can consider the set of updates of the parent window to be included in the state of the parent window.  Then, since the number of updates of the parent is, in expectation, a constant factor more than the size of the child window, the child window can reconstruct the necessary information from the state of the parent and its own set of updates.  
  
  Thus, this algorithm meets the criteria of \Cref{def:divide-and-conquer}, where the expected work at each window \(W\) is 
  \[O(\mathrm{update}(\alg) \cdot |W|).\]
\end{proof}

Now, if we directly apply \Cref{thm:offline-to-fully-dynamic}, we recover an algorithm that does total work
\[O((T + ||\predvec - \delvec||_1) \cdot \mathrm{update}(\alg) \cdot \log^3(T) \log\log(T \cdot |\mathcal{S}|)),\]
when given predictions of \emph{all updates} at the outset of the algorithm.  However, we wish to consider a setting in which the elements arrive arbitrarily, and an element's predicted deletion time is only provided when the element is inserted.  

This leads to two main issues.  The first is that we cannot create a tentative schedule of events during preprocessing, because we have no information about the events that are going to occur.  The second is that we cannot compute a preliminary version of the divide-and-conquer algorithm during preprocessing, because we do not have a tentative schedule of events to refer to.  To address the first issue, we make the following observation.

\begin{observation}[Online scheduling]
Even in the setting where the initial schedule of events is computed statically in the preprocessing step, we use an \emph{online} algorithm to assign events to days.  

Thus, as predicted events arrive over time, we can assign each event irrevocably to an initial tentative schedule using this online algorithm, while maintaining the runtime and approximation guarantees of the original algorithm.    
\label{obs:online-scheduling}
\end{observation}
Thus, we can schedule events as they arrive, without having to reschedule events to account for later arrivals.  

To address the second issue, we show that, while we don't know enough about the schedule to run the entire divide-and-conquer algorithm during preprocessing, we actually can utilize a ``just in time" approach.  

\begin{lemma}[Windows can be computed ``just in time"]
    Consider a offline-dynamic divide-and-conquer algorithm of the form constructed by \Cref{lem:divide-and-conquer-from-incremental}.  Then, the fully-dynamic algorithms with predictions resulting from \Cref{thm:offline-to-fully-dynamic} can be run in a ``just in time" fashion where 
    \begin{enumerate}
        \item the computation associated with a window \(W\) beginning on day \(t\) is only run on or after day \(t\), and 
        \item on day \(t\), we have enough information to run the computation associated with all windows starting on day \(t\), and 
        \item \(\textsc{Retrigger}\) is never called for an insertion event.
    \end{enumerate}
    \label{lem:incremental-just-in-time}
\end{lemma}

\begin{proof}
    First, we note that the output of the fully-dynamic algorithm resulting from \Cref{thm:offline-to-fully-dynamic} on day \(t\) only depends on the computations associated with windows that contain \(t\).  Thus, the first point of the lemma is true for \emph{any} algorithm of this form.  That is, on a day \(t\), we can consider all windows starting after day \(t\) to be ``inactive," and ignore them when they are part of a \(\textsc{Retrigger}\) operation.  Then, on day \(t\), we run the computations of any windows beginning exactly on day \(t\), for the first time.\footnote{In fact, while we don't get a better asymptotic runtime bound, running any algorithm of this form in a ``just in time" way can only reduce the total work done, as we do not have to retrigger windows that are not yet active.}  

    For the second point, consider a window \(W\) beginning on some day \(t\).  It is possible that on day \(t\), there are some events that will occur during \(W\), for which we do not have predicted times.  For example, insertions occur entirely arbitrarily, so we do not have any information about the insertions that will occur during \(W\).  However, any element that could be permanent for \(W\), must 
    \emph{already} be inserted.  This is by definition: for an element to be permanent for \(W\), it must be inserted on or before the first day of \(W\).  Thus, on day \(t\), we have enough information to compute all of the permanent elements of \(W\), which is enough to run the divide-and-conquer algorithms constructed by \Cref{lem:divide-and-conquer-from-incremental}.

    Finally, for the third point, note that any window that depends on an insertion event \(e\) is only computed for the first time after \(e\) occurs.  Thus, at this point, \(e\) is fully known, and will not be rescheduled, so it will never call \(\textsc{Retrigger}\).
\end{proof}

Altogether, \Cref{lem:divide-and-conquer-from-incremental}, \Cref{obs:online-scheduling}, and \Cref{lem:incremental-just-in-time}, along with \Cref{thm:offline-to-fully-dynamic} give us the following theorem. 

\begin{theorem}[Incremental to Fully-Dynamic]
    Given an incremental algorithm, $\alg$, with expected worst-case update time \(\mathrm{update}(\alg)\), we can construct an algorithm in the predicted-deletion model (\Cref{def:predicted-deletion-dynamic-model}), such that the total expected work done by the algorithm is 
    \[O\left( \mathrm{update}(\alg) \cdot (T + ||\mathbf{p} - \mathbf{d}||_1) \cdot \log^4 T \log \log (T\cdot |\mathcal{S}|) \right),\]
    where \(||\mathbf{p} - \mathbf{d}||_1\) is the \(\ell_1\) error of the deletion-time predictions.
\end{theorem}

Using~\cref{thm:offline-to-fully-dynamic}, we obtain the following work, with high probability.

\begin{corollary}
    Given an incremental algorithm, $\alg$, with expected worst-case update time \(\mathrm{update}(\alg)\), we can construct an algorithm in the predicted-deletion model (\Cref{def:predicted-deletion-dynamic-model}), such that the total work done by the algorithm is 
    $\tO\left(\mathrm{update}(\alg) \cdot (T + ||\mathbf{p} - \mathbf{d}||_1)\right)$
    with high probability, where \(||\mathbf{p} - \mathbf{d}||_1\) is the \(\ell_1\) error of the deletion-time predictions.
\end{corollary}

\subsection{Comparison to concurrent work of \cite{PR23, BFNP23}}
\label{sec:comparison-to-vdB-etal}

The independent and concurrent work of \cite{BFNP23} also provides a reduction that lifts an incremental algorithm to the fully dynamic setting given predictions of deletion times.  In this section, we overview their approach and provide a comparison to our reduction.  

The reduction of \cite{BFNP23} is based on a simple and elegant observation over the previous work of \cite{PR23}.  In \cite{PR23}, they consider a model that they call the \emph{deletion look-ahead model}.  In this model, on every day \(t\), the relative deletion order of all elements that are active in the system is known.  Another way to think of this model is that, when an element is inserted, it announces which of the existing elements it will be deleted after.  This is similar to our predicted-deletion model in the regime where the prediction error is 0.  However, it is actually more general, because it only requires \emph{ordinal} information rather than \emph{cardinal} information about the deletion times.  

In the language of our work, we can interpret the \cite{PR23} reduction as doing something similar to constructing the offline divide-and-conquer algorithm from \Cref{lem:divide-and-conquer-from-incremental}, and running it with a just-in-time approach (in the manner of \Cref{lem:incremental-just-in-time}).  This is all with respect to a fixed deterministic partition tree (i.e. a perfectly balanced binary tree).  Indeed, if we had access to cardinal information (i.e. the actual day on which each edge is predicted to be deleted), this would suffice to recover an amortized guarantee for this problem.  

To adapt this to the ordinal setting, \cite{PR23} make the following observation.  Over the course of any window \(W\), at most \(|W|\) elements can be deleted.  Thus, at the start of a window \(W\), any element that is \emph{not} in the \(|W|\) earliest-to-be-deleted elements will certainly be permanent over \(W\),\footnote{Here, we are abusing our notation of permanent elements from earlier, and referring to an element as permanent if it is inserted on or before the first day of \(W\), and deleted on or after the last day of \(W\).} and can safely be inserted.  At the end of the computation associated with \(W\), there are still up to \(|W|\) elements that we did not insert, because we were unsure of whether or not they would be permanent.  These elements are passed to the child windows, which process these elements, along with any elements that were inserted in their sibling windows.  In all, the total number of elements that any window has to process is at most a constant factor more than the size of the window.  

This adaptation to the ordinal setting leaves the reduction of \cite{PR23} with a very useful property: the state of the incremental algorithm at each leaf of the partition tree has elements inserted in \emph{approximately reverse deletion order}.  That is, any element that will be deleted in the next \(d\) deletions, was one of the last \(c \cdot d\) elements to be inserted, where \(c\) is a constant.  

Now, in the case where the predicted deletion order is subject to error, \cite{BFNP23} observe that this property allows us to fix the data structure with low overhead.  In particular, consider an element \(e\) that is deleted earlier than predicted.  That is, on day \(t\), \(e\) was predicted to be the \(d\)-th next element to be deleted, but instead it was the first element to be deleted.  The approximate reverse deletion order property tells us that \(e\) was in the last \(c \cdot d\) elements to be inserted into this leaf.  This means that it was inserted in some window \(W\) relatively close to leaf-level in the partition-tree.  Thus, to fix the data structure, they only need to retrigger these few windows that are descendants of \(W\).  

Essentially, this approximate reverse deletion order property allows the data structure to be fixed on the fly, with low overhead.  Furthermore, their strategy does \emph{not} require a random partition-tree: indeed it is fully deterministic, and can even be adapted to provide a worst-case per update guarantee (as opposed to our work which provides an amortized guarantee).  

In comparison to our work, it is not clear if it is possible to extend the \cite{PR23, BFNP23} approach to lift offline algorithms to the fully 
dynamic setting.  In particular, we provide a simple analogue of their approach that can convert a decremental algorithm to the fully-dynamic setting, when given predictions of insertion times.  We explore this connection formally in the next section (\Cref{sec:decremental}).  With respect to the specific reduction of \cite{BFNP23}, it essentially maintains state of a decremental algorithm, where elements are deleted in \emph{approximately reverse insertion order}. 

While this approach can handle predicted deletion times given an incremental algorithm, and predicted insertion times given a decremental algorithm, it is not clear if and how these two guarantees can be combined.  In particular, this runs into two issues.  
\begin{enumerate}
    \item It is not clear what kind of algorithm to reduce to, as one direction requires an incremental algorithm, and the other direction requires a decremental algorithm.  
    \item It is not clear what the analogue of maintaining elements in both reverse insertion \emph{and} reverse deletion order should be.
\end{enumerate}

Our work addresses the first point by making the key observation that we can reduce to an offline divide-and-conquer algorithm.  As to the second point, we resolve the issue by using a random partition-tree.  This does not require us to maintain any ordering of the elements that we consider. 

Thus, the reduction of \cite{BFNP23} for the predicted deletion model is based on a simple and elegant extension of the previous framework of \cite{PR23}, which lifts incremental algorithms into the deletion look-ahead setting.  The reduction that we present in this work for the predicted deletion model, is instead based on a simple extension from our earlier framework that lifts offline divide-and-conquer algorithms into the predicted update setting. Under a single framework, we generalize all three settings: offline, incremental, and decremental.

\section{Decremental to Fully-Dynamic Transformation}\label{sec:decremental}

In this section, we show how to adapt our framework to convert a decremental algorithm to the fully-dynamic setting, given predictions of only the insertion times.  Formally, we consider the following model. 

\begin{definition}[Predicted-Insertion Dynamic Model]
    In the predicted-insertion dynamic model, we consider a set \(S\) of all the elements that are predicted to ever appear in the system.  At the outset, the algorithm is given \(S\), and predictions for when each of the elements in \(S\) will be inserted.
    Then, on each day, exactly one of the following occurs:
\begin{enumerate}
    \item An element \(e\), either in \(S\) or not, is inserted, 
    \item A previously inserted element is deleted, and provides a prediction of when it will be reinserted (if ever).  
\end{enumerate}
An algorithm computes a function \(f(\cdot)\) in the predicted-insertion dynamic model, if on every day \(t\), the algorithm outputs \(f(S_t)\), where \(S_t\) is 
the \emph{working subset} induced by the true (not predicted) sequence of element insertions and deletions that occur in time-steps $1, \dots, t$.
\label{def:predicted-insertion-dynamic-model}
\end{definition}

A number of works in the past have performed fully dynamic to decremental reductions, starting with the seminal work of Henzinger and King~\cite{HK97}, \emph{for specific problems}. We show that, with the help of predictions, we can generalize such ideas to \emph{any} worst-case 
decremental algorithm.
Now, we show that, given a decremental algorithm, we can design algorithms for the predicted insertion model.  We do this via a simple reduction to the predicted insertion case.  

\begin{theorem}[Decremental to Fully-Dynamic]\label{thm:decremental-to-fully-dynamic}
    Consider a decremental algorithm, $\alg$, which takes time \(\mathrm{initialize}_\alg(T)\) to initialize a state containing up to \(T\) elements, and then has worst-case update time \(\mathrm{update}(\alg)\).  Given such an algorithm, we can construct an algorithm for the predicted-insertion model (\Cref{def:predicted-insertion-dynamic-model}), that does expected total work 
    \[O \left( \mathrm{initialize}_\alg(T) \cdot K + \mathrm{update}(\alg) \cdot (T + ||\mathbf{p} - \mathbf{i}||_1 + TK) \log^4(T) \log \log (T \cdot |\mathcal{S}|) \right) ,\]
    where \(||\mathbf{p} - \mathbf{i}||_1\) is the \(\ell_1\) error of the insertion-time predictions for elements in \(S\), and \(K\) is the number of elements that are not in \(S\) that are ever inserted. 
\end{theorem}

\begin{proof}
    We reinterpret $\alg$ as in incremental algorithm on ``anti-elements."  That is, suppose we initialize our algorithm to contain the elements in a set \(S\).  Then, as a decremental algorithm, the algorithm can handle updates that delete elements of \(S\).  We can also think of this as an incremental algorithm, that can handle the insertion of anti-elements, corresponding to elements in \(S\).

    Now, consider a series of updates in the predicted-insertion model (\Cref{def:predicted-insertion-dynamic-model}).  From the view of our decremental algorithm, each element starts as not being in the system, then is inserted at some time for which we are given a prediction, then is deleted at an unknown arbitrary time, when it predicts its next insertion.  
    
    From the view of the incremental algorithm, each anti-element starts \emph{in} the system, then is deleted at some time for which we have a prediction, and then is reinserted again at an unknown arbitrary time, at which point it predicts its next deletion.

    Thus, viewing the decremental algorithm as an incremental algorithm on anti-elements, makes it fit directly into our earlier framework for incremental algorithms, because it converts the decremental algorithm into an incremental algorithm, and it converts insertion events into deletion events and vice versa. 

    The only remaining issue is that this strategy cannot handle the case where an element that was never part of the predicted set \(S\) is inserted.  In this case, we must add the new element to \(S\), and run the initialization of the algorithm again from scratch.  This contributes total work 
    \[O \left( \mathrm{initialize}_\alg(T) \cdot K \right),\]
    where \(K\) is the number of such elements outside \(S\) that are ever inserted.
    Then, we must retrigger the entire tree, which is equivalent to this new element contributing \(\ell_1\)-error of \(T\).  Over all such elements this contributes work 
    \[O \left( TK \cdot \log^4 T \log \log (T \cdot |\mathcal{S}|) \right) .\]

    Adding these contributions to the work bound from the incremental reduction, gives us our final bound of 
    \[O \left( \mathrm{initialize}_\alg(T) \cdot K + \mathrm{update}(\alg) \cdot (T + ||\mathbf{p} - \mathbf{i}||_1 + TK) \log^4 T \log \log (T \cdot |\mathcal{S}|) \right) .\]
\end{proof}

Using~\cref{thm:offline-to-fully-dynamic}, we obtain the following work, with high probability.

\begin{corollary}\label{cor:decremental-to-fully-dynamic}
    Consider a decremental algorithm, $\alg$, which takes time \(\mathrm{initialize}_\alg(T)\) to initialize a state containing up to \(T\) elements, and then has worst-case update time \(\mathrm{update}(\alg)\).  Given such an algorithm, we can construct an algorithm for the predicted-insertion model (\Cref{def:predicted-insertion-dynamic-model}), that does total work 
    $\tO \left( \mathrm{initialize}_\alg(T) \cdot K + \mathrm{update}(\alg) \cdot (T + ||\mathbf{p} - \mathbf{i}||_1 + TK) \right)$, with
    high probability,
    where \(||\mathbf{p} - \mathbf{i}||_1\) is the \(\ell_1\) error of the insertion-time predictions for elements in \(S\), and \(K\) is the number of elements that are not in \(S\) that are ever inserted. 
\end{corollary}

We make the following remark that allows us to obtain a deterministic, worst-case decremental to fully-dynamic transformation using~\cite{BFNP23}.

\begin{remark}
    The above reduction is not specific to our particular implementation of the incremental to fully dynamic reduction.  In particular, our interpretation of anti-elements can be applied to the incremental to fully-dynamic reduction in the concurrent independent work of \cite{BFNP23} and achieve their deterministic worst-case per update work bound also in the decremental setting.  
\end{remark}

\section{Applications: Offline, Incremental, and Decremental to Fully Dynamic}\label{sec:problems}

We apply our framework to the following problems to obtain fully dynamic algorithms in the predicted-deletion model. A summary of our 
runtimes compared to the best-known fully dynamic algorithms is given in~\cref{table:problems}.
We define the problems we study below. Unless specified, the input to the following problems is an input graph $G = (V, E)$.

\begin{itemize}
    \item \textbf{All-Pairs Shortest Paths (APSP):} For any pair of queried vertices $u$ and $v$, determine the distance from $u$ to $v$. In the planar diagraph APSP problem, we are given a planar, directed and weighted graph. A $c$-approximate APSP algorithm returns an approximation that is within $c$-factor of the 
    distance between the pair.
    \item \textbf{Triconnectivity:} For any pair of queried vertices, determine whether the pair is \emph{3-vertex connected}. Two pairs of vertices $u$ and $v$
    are $3$-vertex connected if and only if $u$ and $v$ remain connected whenever fewer than two vertices are removed.
    \item \textbf{Dynamic Depth-First Search (DFS) Tree:} A dynamic DFS tree is a DFS tree that is reported after each edge insertion or deletion and 
    is a valid DFS tree for the current graph. 
    \item \textbf{All-Pairs Maximum Flow:} For any pair of queried vertices, $s$ and $t$, return the maximum flow between $s$ and $t$. A $c$-approximate 
    maxflow algorithm returns a flow value that is a $c$-factor approximation of the actual maxflow.
    \item \textbf{All-Pairs Minimum Cut:} For any pair of queried vertices, $s$ and $t$, return the minimum cut between $s$ and $t$. A $c$-approximate 
    maxflow algorithm returns a cut value that is a $c$-factor approximation of the actual min-cut.
    \item \textbf{Multi-Commodity Concurrent Flow:} For a set of triples $\{(s_1, t_1, d_1), \dots, (s_P, t_P, d_P)\}$, return the maximum 
    value $\alpha$ where, concurrently 
    for all $i \in [P]$, $s_i$ can send $\alpha \cdot d_i$ units of flow to $t_i$. A $c$-approximate multi-commodity concurrent flow algorithm returns a value 
    that is at least $\alpha/c$. 
    \item \textbf{Uniform Sparsest Cut:} For a given graph $G = (V, E, \mathbf{w})$, return $\Phi_G = \min_{S \subset V} \frac{\sum_{(u, v) \in E, u \in S, v \not\in S} \mathbf{w(u, v)}}{|S| \cdot |V \setminus S|}$.
    \item \textbf{Monotone Submodular Maximization:} Given a ground set $N = [n]$ and a set function $f: N \rightarrow \mathrm{R}^{+}$, function $f$ is monotone
    if $f(A) \geq f(B)$ for any $B \subseteq A \subseteq N$. Function $f$ is submodular if $f(A \cup \{u\}) - f(A) \leq f(B \cup \{u\}) - f(B)$ for any
    $B \subseteq A \subseteq N$ and element $u$. Under cardinality constraint $k$, the problem maximizes $\max_{S \subseteq [n], |S| = k} \left(f(S)\right)$
    for some parameter $1 \leq k \leq n$ and under a matroid constraint $\mathcal{M}$, the problem maximizes $\max_{S \subseteq \mathcal{M}} \left(f(S)\right)$.
    In the dynamic setting, elements can be inserted and deleted from the ground set, and the goal is to maintain $\max_{S \subseteq \mathcal{M}}\left(f(S)\right)$. A $c$-approximate algorithm
    returns a set $S$ that has value at least $c \cdot \max_{S \subseteq \mathcal{M}}(f(S))$.
    \item \textbf{$k$-Edge Connectivity:} For any pair of queried vertices $\{u, v\}$, the pair $\{u, v\}$ is $k$-edge connected if and only if
    $u$ and $v$ remain connected whenever any set of $k-1$ edges are removed.
    \item \textbf{Minimum Spanning Forest (MST):} A dynamic minimum spanning forest is an minimum weight spanning forest that is maintained under edge insertions and deletions.
    \item \textbf{Strongly Connected Components (SCC)/Topological Sort:} Given a directed input graph $G = (V, E)$, maintain the strongly 
    connected components and topological sort, respectively, under edge insertions/deletions.
\end{itemize}

\paragraph{Offline to Fully Dynamic}

In this section, we instantiate several offline-to-fully dynamic algorithms using our~\cref{thm:offline-to-fully-dynamic}. Specifically, we show
transformations for triconnectivity, $k$-edge connectivity, and minimum spanning forest using the divide-and-conquer algorithms of~\cite{chalermsook2021vertex,eppstein1994offline,PSS17}, achieving exponential time improvements on the running times of the best-known 
fully dynamic algorithms when $||\predvec - \realvec||_1 = \tO(T)$ where $\predvec$ is the predicted sequence of update times and $\realvec$
are the real update times.

We first prove that the offline algorithms given in~\cite{chalermsook2021vertex,eppstein1994offline,PSS17} falls under our definition of
$c$-divide-and-conquer algorithms and specify the parameters for each algorithm. After proving these lemmas, our framework directly gives 
fully dynamic algorithms for these problems given predicted-updates.

\begin{lemma}\label{lem:triconnectivity-offline}
    There exists a $c$-divide-and-conquer algorithm for offline triconnecitivity~\cite{PSS17} where $c = 1$.
\end{lemma}

\begin{proof}
    \cite{PSS17} provides a divide-and-conquer algorithm on the sequence of events where each subproblem is half of its parent. In each subproblem,
    they produce a $3$-vertex sparsifier with size equal to the number of non-permanent edges in the subproblem. The permanent edges in each
    subproblem are used to contract the graph to smaller sizes. For a subproblem $S$ of size $|S|$ (where $S$ contains the non-permanent edges),
    the sparsification can be done in time $O(|S|)$.
    Thus,~\cite{PSS17} gives a $c$-divide-and-conquer algorithm where $c = 1$.
\end{proof}

\begin{lemma}\label{lem:c-edge-offline}
    There exists a $c$-divide-and-conquer algorithm for offline $k$-edge connectivity~\cite{chalermsook2021vertex} for any constant $k\geq 1$ where 
    $c = 1$.
\end{lemma}

\begin{proof}
    \cite{chalermsook2021vertex} provides a divide-and-conquer algorithm that divides the events in approximately equal sizes. 
    Then, in each subproblem $S$ of size $|S|$ (where $S$ contains the non-permanent edges), they construct a vertex sparsifier of 
    size $\tO(|S|)$ with the terminals being the set of queried vertices falling within the subproblem. Each sparsifier can be constructed
    in $\tO(|S|)$ time; hence,~\cite{chalermsook2021vertex} gives a $c$-divide-and-conquer algorithm where $c = 1$.
\end{proof}

\begin{lemma}\label{lem:mst-offline}
    There exists a $c$-divide-and-conquer algorithm for offline minimum spanning forest~\cite{eppstein1994offline} where $c = 1$.
\end{lemma}

\begin{proof}
    \cite{eppstein1994offline} provides a divide-and-conquer algorithm that recursively divides the sequence into halves. 
    Then, in each subproblem $S$ where $|S|$ denotes the number of non-permanent edges, %
    they sparsify using the permanent and non-permanent edges in the following way. 
    They run the MST algorithm by successively picking edges in non-decreasing weight 
    using a standard MST algorithm like Kruskal's in $O(|S|\log(|S|))$
    time, \emph{breaking ties by giving preference to non-permanent edges}. All non-permanent edges that are not picked
    are deleted and every picked permanent edge is contracted. This results in a sparsifier that has size $O(|S|)$
    that is passed to children. Hence,~\cite{eppstein1994offline} gives a $c$-divide-and-conquer algorithm where $c = 1$.
\end{proof}

\cref{lem:triconnectivity-offline,lem:c-edge-offline,lem:mst-offline} combined with~\cref{thm:offline-to-fully-dynamic}
directly give the following theorem.

\begin{theorem}[Predicted-Updates Algorithms]\label{thm:predicted-updates-apps}
    Using our framework given in~\cref{alg:high-probability}, 
    we obtain the following fully dynamic algorithms in the predicted-updates model, 
    assuming $||\predvec - \realvec||_1 = O(T)$, where $T$ is the total number of updates, 
    $\predvec$ is a vector of predicted event times, and $\realvec$ is a vector
    of real event times:
    \begin{enumerate}
        \item An algorithm for triconnectivity in $\tO\left(1\right)$ amortized update and query times; %
        \item An algorithm for $k$-edge connectivity in $\tO\left(1\right)$ amortized update time and query time %
        for any constant $k \geq 1$; and
        \item A dynamic minimum spanning forest maintained in $\tO\left(1\right)$ amortized update time.
    \end{enumerate}
\end{theorem}

\paragraph{Incremental to Fully Dynamic}
To obtain our fully dynamic algorithms, we use the worst-case incremental algorithms given in the following works~\cite{baswana2019dynamic,ChenDWZZ18,chen2020fast,das2022near,FeldmanLNSZ22,holm2020worst,Goranci19,GHS19}
combined with~\cref{cor:update-query-work} to obtain~\cref{thm:app-runtimes}.

\begin{theorem}[Predicted-Deletion Algorithms]\label{thm:app-runtimes}
    Using our framework, we obtain the following fully dynamic algorithms in the predicted-deletion model, assuming $||\predvec - \delvec||_1 = O(T)$, where $T$ is the total number of updates, $\predvec$ is a vector of predicted deletion times, and $\delvec$ is a vector
    of real deletion times:
    \begin{enumerate}
        \item An exact all-pairs shortest path algorithm for planar directed graphs with $\tO(\sqrt{n})$ 
        amortized update time and $\tO(\sqrt{n})$ worst-case query time;
        \item An algorithm for triconnectivity in $\tO\left(1\right)$ amortized update time and $O(\log^3(n))$ worst-case query time;
        \item A dynamic DFS tree reported in $\tO(n)$ amortized update time;
        \item A $O\left(\log^{8k}(n)\right)$-approximate maxflow/min-cut algorithm with $\tO\left(n^{2/(k+1)}\right)$ amortized update time and reports the maxflow between 
        any pair of vertices $s$ and $t$ in $\tO\left(n^{2/(k+1)}\right)$ worst-case query time;
        \item A $O\left(\log^{8k}(n)\right)$-approximate multi-commodity concurrent flow algorithm with $\tO\left(n^{2/(k+1)}\right)$ amortized update 
        time and $\tO(P^2)$ worst-case query time where $P$ is the number of queried pairs;
        \item A $O\left(\log^{8k}(n)\right)$-approximate uniform sparsest cut algorithm with $\tO\left(n^{2/(k+1)}\right)$ amortized update time.
        \item A $0.3178$-approximate monotone submodular maximization algorithm under a matroid constaint of rank $k$ and makes $\tO(\poly(k, \log n))$
        function evaluations per update for any $n, k > 0$. 
    \end{enumerate}
\end{theorem}

\paragraph{Decremental to Fully Dynamic}

To obtain our fully dynamic algorithms, we use the worst-case decremental algorithms for 
strongly connected components of Roddity~\cite{10.5555/2627817.2627899} and
topological sort (trivial)~\cref{cor:update-query-work} to obtain~\cref{thm:decremental-to-fully-dynamic} when
$||\predvec - \insvec||_1 = O(T)$ where $\predvec$ are the predicted times for insertions, $\insvec$
is the vector of real insertion times, and $T$ is the total number of updates.

\begin{theorem}[Predicted-Deletion Algorithms]\label{thm:app-runtimes-deletion}
    Using our framework, we obtain the following fully dynamic algorithms in the predicted-deletion model, assuming $||\predvec - \insvec||_1 = O(T)$, where $T$ is the total number of updates, $\predvec$ is a vector of predicted insertion times, and $\insvec$ is a vector
    of real insertion times:
    \begin{enumerate}
        \item An algorithm for maintaining strongly connected components with $\tO(m)$ 
        amortized update time time (matching fine-grained lower bounds of~\cite{abboud2014popular}); and
        \item An algorithm for maintaining a topological sort
        in $\tO\left(1\right)$ amortized update time.
    \end{enumerate}
\end{theorem}

\section{Future Directions}

\paragraph{Extension to other problems.}  
One follow-up direction to our work is to find more problems to which this framework can be applied.  This could help design algorithms that are more efficient and/or simpler than the state-of-the-art fully-dynamic solution.
Alternatively, it would be interesting to see algorithms designed for this model that look very different from this framework.  
A particularly interesting follow-up question is: are there problems for which the a predicted-deletion dynamic algorithm can circumvent a fully-dynamic lower bound when the prediction error is sufficiently low?

\paragraph{Derandomization and/or robustness to adaptive adversary.}  The framework presented in this work has a key reliance on randomization and is not robust to an adaptive adversary.  It is an interesting question whether the framework can be derandomized, or otherwise made robust to adaptive adversaries.  Potential challenges and avenues, along with vulnerabilities of the framework in this work to an adaptive adversary are discussed in depth in \Cref{sec:use-of-randomness}.

\paragraph{Removing amortization.}  The framework presented in this work relies heavily on amortization to achieve the desired runtime bounds.  It is an interesting question of whether problems in this model can meet a per-update runtime bound.  
It is also not immediately clear what role prediction error should play in such a bound.  For example, one could hope to design an algorithm that, on each day \(t\), has the update time of an incremental algorithm, with an additive factor that scales with the ``error at day \(t\)."  This would require an interesting definition of error in a local sense, in contrast to the global measure of error that we consider in this work (\(\ell_1\) error).

\paragraph{Lower bounds.}  It would be interesting to get fine-grained lower bounds for problems in this model.  In particular, this framework accrues logarithmic factors in various places, including a competitive ratio for online metric matching on a line, and the expected size of the smallest node of the partition tree fully containing a subsequence \([a, b]\).  For the approach in this work, some of the logarithmic factors are necessary side effects of lower bounds for these quantities.  It would be interesting to see if there are indeed necessary for any approach to solve dynamic problems in this model, or if they can be circumvented by a different approach.  

\paragraph{Implementation.} A main motivation of this model is that it could potentially allow practitioners to take advantage of known structure in real-world data that classical fully-dynamic algorithms are oblivious to.  It would be interesting to see how algorithms designed in this framework compare with the best algorithms and heuristics in practice.

\addcontentsline{toc}{section}{Acknowledgements}
\section*{Acknowledgements} We are very grateful to Gramoz Goranci and anonymous reviewers for providing very helpful references for application problems we address in this paper.  We are also grateful to Pat Sukprasert for pointing us to references on online bipartite matching.  

\appendix
\section{Appendix}\label{appendix}

\subsection{Connection to Metric Tree Embeddings}\label{app:metric-embedding}

The local error guarantee of our random partition-tree essentially provides a randomized scheme to embed a particular line metric where the points are at integer increments, into a tree metric of low depth. Precisely, we want a low-depth partition-tree such that, for two points \(a, b\), the size (number of descendant leaves) of the lowest common ancestor of \(a\) and \(b\) is of size \(O(|a-b| \cdot \log T)\) in expectation, where \(T\) is the number of points in the metric space.  This is morally the same as finding a low-depth tree embedding of this metric space.  One direction is seen by setting edge lengths of a low-depth partition-tree such that for a parent node \(p\) and its child node \(c\), \(\mathrm{dist}(p, c) = \frac{1}{2}(|p| - |c|)\).  Thus, the distance between two leaves is precisely the size of the lowest common ancestor in the partition-tree, and this gives a low-depth tree embedding of the space.  The other direction is seen by noting that for any embedding of this particular metric space into a low-depth hierarchically well-separated tree (HST), e.g. \cite{FRT03}, the distance between two points \(a, b\) in the tree is at least the size of the lowest common ancestor in the tree, and thus the HST gives us a low-depth partition-tree. 
A line metric is, of course, itself a tree metric, and one can therefore find a ``tree embedding" with no distortion.  However, for our algorithm, we require an embedding into a tree metric that is low-depth.  This can also be achieved by other distributions on trees, e.g. the classic tree-embedding scheme of \cite{FRT03}.  

\subsection{Use of Randomness}
\label{sec:use-of-randomness}

One implication of the way we use randomness in our algorithmic framework is that we \emph{do not} expect this framework to be robust to the
strongest form of an \emph{adaptive adversary} where the 
adversary has access to the random bits used within our algorithm. 
(Although in many learning-augmented models, the definition of an adaptive adversary is often also not immediately clear.)  
Thus, it is an interesting question whether such a result can be derandomized, as a deterministic algorithm must be robust to an adaptive adversary.  

There are two places where we use randomness in our predicted-deletion framework.
\begin{enumerate}[(1)]
    \item Random partition-tree: this allows us to argue that in effect that two days that are close in time, correspond to leaves in the partition-tree that are close in the tree, in expectation.  
    \item Online metric matching: the randomized harmonic algorithm for online metric matching can be implemented quickly.
\end{enumerate}

We discuss challenges and potential avenues to derandomize these components, along with a strategy that is open to an adaptive adversary.  

\paragraph{Random partition-tree.}  The main benefit of a random partition-tree is that it preserves lengths in expectation.  That is, \Cref{lem:random-tree-preserves-lengths} tells us that for \(a, b \in [T]\), the expected size of the lowest common ancestor of \(a\) and \(b\) in a random-partition-tree drawn over \([T]\) is \(O(|b - a| \log T)\).  

It is worth noting that this guarantee does \emph{not} hold with high probability.  In fact, the probability that the size of the lowest common ancestor is \(T\) (the root node), is \(\frac{|b - a|}{T - 1}\) (i.e. this is the probability that the first divider chosen in the tree falls in the range \([a, b]\)).  

This indicates that this expectation bound is morally fulfilling two functions: it is both avoiding some bad cases, and implicitly accounting for some amortization.  
To illustrate this, consider the following examples. 

\begin{example}[Uneven error density]
Let \(T\) be \(2^i\) for some even \(i\).  Consider \(\sqrt{T}\) elements that are added on days \(1, 2, \dots, \sqrt{T}\) with predicted deletion times 
\[T/2 + 1, T/2 + 2, \dots, T/2 + \sqrt{T} .\]
Their actual deletion times will be 
\[T/2 + 1 - \sqrt{T}, T/2 + 2 - \sqrt{T}, \dots, T/2 .\]  
The contribution to the \(\ell_1\) prediction error is \(\sqrt{T}\) for each of \(\sqrt{T}\) elements, for a total of \(T\).  
\label{ex:uneven-error-density}
\end{example}

The deterministic perfectly balanced partition-tree will make the first split between \(T/2\) and \(T/2 + 1\).  So the lowest common ancestor of the predicted deletion time and the actual deletion time for each of these elements is the root node.  Thus, each of these early deletion events will trigger \(O(T)\) recomputation, for a total of \(O (T \sqrt{T}) \) work.  (We are being informal about the exact logarithmic factors here, for illustrative purposes.)

The randomized tree, on the other hand, guarantees that the expected work of triggering recomputations for these events is \(O (T \log T)\).  This is only a \(\log T\) factor off from the actual \(\ell_1\) error, as opposed to the deterministic tree that is a \(\sqrt{T}\) factor off.  

\Cref{ex:uneven-error-density} illustrates that if there is a portion of the sequence of days that is has a relatively high density of error that coincides with a ``deep split" of our tree, then the work of our algorithm could be high.  The random partition-tree allows us to avoid this by placing the ``splits" of the tree randomly.  

Now, we consider the counterpart problem. 

\begin{example}[Even error density]
    Let \(T = 2i\) for some \(i\).  Label the days \(\{1, \dots, T\}\) Consider \(i\) elements that such that element \(e_j\) for \(j \in \{0, \dots, i - 1\}\) is inserted on day \(2j + 1\) with predicted deletion \(2j + 5\).  \(e_j\)'s actual deletion time is \(2j + 2\).   The last elements \(e_{i - 2}\) and \(e_{i - 1}\) alone will have predicted deletion on the actual deletion days of \(T - 2\) and \(T\), respectively.  

    The contribution to the \(\ell_1\) prediction error is \(3\) for each of \(\frac{T}{2} - 2\) elements, for a total error of \(\frac{3T}{2} - 6 \in O(T)\).
    \label{ex:even-error-density}
\end{example}

In this example, except for the first and last possible divider, every other divider splits some actual deletion time from its corresponding deletion time.  
This says that in a deterministic partition-tree, and in most instances of a randomized partition-tree, there is some error of size 3 that triggers \(O(T)\) recomputation.  

While this is a large blow up, this does not occur for \emph{most} of the errors, and specifically for \Cref{ex:even-error-density} one can show that both the deterministic tree and the randomized tree will trigger \(O(T)\) recomputation total, which is comparable to the size of the error, \(O(T)\).
This demonstrates the implicit amortization that happens when we consider the \emph{expected} cost associated with a specific error in the randomized partition-tree. 

These two examples demonstrate that a deterministic scheme could likely achieve the amortization of the random partition-tree and handle inputs like \Cref{ex:even-error-density}.  The challenge lies in designing a deterministic scheme that can avoid splitting areas of high error density, as in \Cref{ex:uneven-error-density}.

This also suggests a strategy that an adaptive adversary could use against a random partition-tree that would break the guarantee, even if the adversary does not have direct access to the random tree.  An adversary could learn about the structure of the tree from seeing the output solutions of the framework.  As an example, consider a divide-and-conquer algorithm that computes some greedy solution over permanent elements.  From seeing the daily outputs of the divide-and-conquer algorithm, the adversary could learn about the order in which the elements were given to the divide-and-conquer algorithm.  This is highly informative about what level of the tree each element becomes permanent at, and since the adversary knows the input, they learn about where the windows of each level end.  Then, to make the algorithm incur high cost, the adversary could choose to delete elements just before ``deep splits" in the tree.  

\paragraph{Online metric matching.}  For this framework, we want an algorithm for online metric matching on the line metric that achieves a logarithmic (or polylogarithmic) competitive ratio, and can be implemented in polylogarithmic update time.  

In the online metric matching problem, the algorithm is given a set of ``servers".  Then, ``requests" arrive online and report their distances (cost to be matched) to each of the servers. For each request that arrives, the algorithm must immediately match it irrevocably to one of the servers, and that server may not be matched to any future requests.  The objective function is the total cost of the final matching.  The distance function must be consistent with some metric space.  

The online metric matching problem was introduced independently by Khuller, Mitchell, and Vazirani \cite{KMV91}, and by Kalyanasundaram and Pruhs \cite{KP93}.  Both of these works show a deterministic \((2n - 1)\) competitive algorithm, where \(n\) is the number of points in the metric space.  This is also optimal for deterministic algorithms, there is no deterministic algorithm that achieves competitive ratio better than \((2n - 1)\) in all metric spaces.  
It is possible to circumvent this lower bound using randomization.  The best known randomized algorithm for general metric spaces is due to Bansal, Buchbinder, Gupta, and Naor, and achieves a competitive ratio of \(O(\log^2 n)\) \cite{BBGN07}.

For line metrics specifically, it is possible to do better.  On the randomized side, Gupta and Lewi provide multiple randomized algorithms achieving competitive ratio \(O(\log n)\) \cite{GL12}, one of which we use in our construction.  On the deterministic side, Raghvendra shows that the deterministic \emph{robust matching algorithm} works well in a variety of settings, and in particular achieves competitive ratio \(O(\log n)\) for line metrics \cite{Rag16, Rag18}.  
In terms of competitive ratio, these results are close to optimal, as Peserico and Scquizzato show that no algorithm, deterministic or randomized, can achieve competitive ratio \(o(\sqrt{\log n})\) for the line metric \cite{PS21}.  

From the perspective of our framework, we could hope to use Raghvendra's robust matching algorithm to achieve the desired competitive ratio deterministically.  However, it is not immediately clear how to implement this algorithm with polylogarithmic update time.  The algorithm maintains an offline matching that informs the choices for the online matching.  The offline matching is maintained by finding an augmenting path on the associated flow network on each step. 
Without additional assumptions, finding such an augmenting path can take \(O(n^2)\) work.  If we were to use the robust matching algorithm as is, it would actually be the dominating factor of our runtime.  
However, the robust matching algorithm as stated can be applied to any metric space.  It is an interesting question of whether, for the line metric specifically, the robust matching algorithm could be implemented much faster.  

We note that it is possible to implement deterministic algorithms for our setting that achieve slightly relaxed guarantees.  In particular, the greedy algorithm (always assigns a request to the nearest available server) achieves the following guarantee for our problem.  

\begin{restatable}[Greedy matching guarantee]{lemma}{greedymatchingguarantee}
\label{lem:greedy-matching-guarantee}
    Given an element insertion $e$ with predicted deletion time $\trep$, we assign $e$ greedily to the nearest free time slot 
    $\widehat{\trep}$. Such an algorithm produces an assignment of elements to time slots such that \[\left(\sum_{i = 1}^t 
    |\trep^i - \widehat{\trep^i}|\right) \leq 4t \cdot \maxe \cdot \log(T) + T \cdot \maxe,\]
    where \(\maxe\) is the minimum \(\ell_\infty\) error of any feasible solution.  That is, over all possible assignments \(t_{\text{baseline}}\) of deletion times, such that \(t_{\text{baseline}}^i \neq t_{\text{baseline}}^j\) for \(i \neq j\), 
    \[\maxe = \min_{t_{\text{baseline}}} \max_{i} \left|t_{\text{baseline}}^i - \trep^i \right|.\]
\end{restatable}

This is proven in \Cref{sec:deterministic-l-inf-bound}.  Note that we leverage the integral structure of our particular problem to get this guarantee (i.e. the requests and server are at integral positions, and there is exactly one server per day.)  This algorithm can be implemented quickly using a union-find data structure in the same way as the harmonic algorithm (see proof of \Cref{lem:preprocessing-schedule}.)

In particular, if we were to use this in place of randomized metric matching in our framework, we would achieve expected total work 
\[O \left( T \left(1 + ||\mathbf{p} - \mathbf{d}||_\infty \log^2 T \right) \cdot \mathrm{update}(A)  \cdot \log^2 T \cdot \log \log (|\mathcal{S}| T) \right),\]
and the only remaining randomized component would be the random partition-tree.  Note that \(T \cdot ||\mathbf{p} - \mathbf{d}||_\infty \) is always at least \(||\mathbf{p} - \mathbf{d}||_1\), and in many instances is much larger.

\subsection{Preprocessing Proofs}\label{app:preprocess}

In this section, we provide the proofs for preprocessing that we omitted in the main body of our paper. We first recall~\cref{lem:preprocessing-schedule}.

\initial*
\begin{proof}
    We iterate through the events in $P$ in arbitrary order; the reason is that the online metric matching algorithm can handle requests that occur in any order. 
    At each predicted event, we match it to an available day assignment, where only one event can be scheduled on each day.  
    Then, this problem becomes an instance of online metric matching, where the predicted events are requests, 
    and there is one available server at each day.  In particular, our problem instances are over the line metric.

    We implement the harmonic algorithm for online metric matching on line metrics, due to Gupta and Lewi \cite{GL12}.  In this algorithm, a request is matched to either its closest open server on the left, or the closest open server on the right.  Let \(d_\text{left}\) be the distance from the request to the closest open server to the left, and \(d_\text{right}\) be the distance from the request to the closest open server on the right.  Then the request is assigned to the left server with probability \(\frac{1/d_\text{left}}{1/d_\text{left} + 1/d_\text{right}}\), and to the right server otherwise. 

    We note that this algorithm can be implemented in amortized \(O(\log^*(T))\) time per update, as shown in \Cref{alg:random-online-matching}.  We maintain a union-find data structure that keeps track of contiguous blocks of assigned servers where initially each $t \in [T]$ is assigned to a different disjoint set.  
    For each block, we maintain metadata with the first open server to the left and the first open server to the right.  Then, when a new request arrives, finding the closest open servers can be done by looking up the request's representative element in the union-find data structure.  One request is assigned per arrival, so at most three sets must be merged for each arrival.  Since union-find can be implemented with amortized \(O\left(\log^* (T)\right)\) update time, and we use a constant number of calls to union-find (with constant additional work), 
    we can implement this algorithm with amortized \(O\left(\log^* (T)\right)\) update time.

    Now we argue that it meets our correctness guarantees.  The harmonic algorithm is \(O(\log(T))\) competitive, for \(T\) requests.  Let \(\realvec\) be the vector of the real
    update days of each event, \(\predvec\) be the predicted update times of each event, and \(\assignedvec\) be the assigned update days of each event.  We have 
    \begin{align*}
        ||\assignedvec - \realvec||_1 &\le ||\assignedvec - \predvec||_1 + ||\predvec - \realvec||_1 \\
        &\le O(\log (T)) \cdot ||\mathbf{c} - \predvec||_1 + ||\predvec - \realvec||_1, 
    \end{align*}
    where \(\mathbf{c}\) is the feasible assignment with the lowest \(\ell_1\) distance to \(\predvec\), in hindsight.  Since we know that \(\realvec\) is also a feasible assignment, we have that 
    \begin{align*}
        ||\mathbf{c} - \mathbf{p}||_1 &\le ||\mathbf{p} - \realvec||_1
    \end{align*}
    Thus, 
    \begin{align*}
        ||\assignedvec - \realvec||_1 &\le O(\log (T)) \cdot ||\predvec - \realvec||_1.
    \end{align*}

    Finally, to resolve the ordering constraints among events on the same element $e$, we iterate through all of our assigned days and maintain a set of all elements for which 
    we have seen a deletion event but not an insertion event. Suppose that for element $e$ we see an insertion event on day $t_{e, i}$ \emph{after} its deletion event on day
    $t_{e, d}$, we reassign the day for the deletion event to $t_{e, i}$. If there are multiple deletion events that occur before an insertion event, we reassign 
    each deletion event to the next available insertion event and to $T + 1$ otherwise. 
    Since we are guaranteed at most one event per day before our reassignment of events and our assignment procedure ensures that we reassign at most one deletion 
    event for each element $e$ to its corresponding insertion event, the maximum number of events assigned to each day $t \in [T]$ is two and there is at most one insertion 
    event and one deletion event assigned to $t$. We produce a new vector $\assignedvec'$ after this final processing.

    We now prove that our new $\assignedvec'$ also satisfies our error bounds. We show the following inequality since all insertion events are assigned the same days as $\assignedvec$
    and only deletions are potentially reassigned:
    \begin{align*}
        ||\assignedvec' - \realvec||_1 &\leq ||\assignedvec - \realvec||_1 + ||\mathbf{a}'_{0, \text{del}} - \realvec_{\text{del}}||_1,
    \end{align*}

    where $\mathbf{a}'_{0, \text{del}}$ and $\realvec_{\text{del}}$ are the assigned and real days, respectively, of the deletion events. We now show that 
    $||\mathbf{a}'_{0, \text{del}} - \realvec_{\text{del}}||_1 \leq ||\assignedvec - \realvec||_1$ through several cases. Let $\hat{t}_{i, e, 1} < \hat{t}_{i, e, 2} < \cdots$ 
    be the insertion
    times for element $e$ in $\assignedvec$, $\hat{t}_{d, e, 1} < \hat{t}_{d, e, 2} < \cdots$ be the deletion times for element $e$ in $\assignedvec$,
    $t_{i, e, 1} < t_{i, e, 2} < \cdots$ be the insertion times for element $e$ in $\realvec$, and $t_{d, e, 1} < t_{d, e, 2} < \cdots$ be the deletion 
    times for element $e$ in $\realvec$. We assign $t_{i, e, j} = T + 1$ for any updates where we have a greater number of predicted events than real events. 
    Then, we show, for any $j$:

    \begin{enumerate}
        \item If $\mathbf{\hat{t}_{i, e, j} < t_{i, e, j} < t_{d, e, j}}$ or $\mathbf{t_{i, e, j} < \hat{t}_{i, e, j} < t_{d, e, j}}$, then $|t_{d, e, j} - \hat{t}_{i, e, j}| < 
        |\hat{t}_{d, e, j} - t_{d, e, j}|$. 
        \item Otherwise, if $\mathbf{t_{i, e, j} < t_{d, e, j} < \hat{t}_{i, e, j}}$, then $|t_{d, e, j} - \hat{t}_{i, e, j}| < 2 \cdot |\hat{t}_{i, e, j} - t_{i, e, j}|$.
    \end{enumerate}

    Hence, via the casework above, we show that 
    
    \begin{align*}
        ||\assignedvec' - \realvec||_1 \leq 2 \cdot ||\mathbf{a}_{0, \text{ins}} - \realvec_{\text{ins}}||_1 + ||\mathbf{a}_{0, \text{del}} - \realvec_{\text{del}}||_1 
        \leq 2 \cdot ||\assignedvec - \realvec||_1,
    \end{align*}

    where $\mathbf{a}_{0, \text{ins}}$ and $\realvec_{\text{ins}}$ are the assigned and real days, respectively, of the insertion events. This proof can be extended beyond insertions 
    and deletions to any constant number of ordering constraints of events on elements. Such a reassignment to produce $\assignedvec'$ requires $O(T)$ additional work.
\end{proof}

\begin{algorithm}
\caption{Fast Randomized Online Matching}
\label{alg:random-online-matching}
\begin{algorithmic}[1]
\Require{Predicted sequence of updates $P$.}
\Ensure{Sequence of updates $\matchP$ with at most two events on each day.}
\State Initialize a union-find data structure over \(T = \max\left(|P|, \max_{(e, type, t, i) \in P}(t)\right)\) elements where the structure has procedures \find and \union
and all $t \in [T]$ are assigned to disjoint sets.
\For {\(t \in T\)}
    \State \(\textbf{left}_t := t - 1\) if \(t > 1\), \(\texttt{null}\) else.
    \State \(\textbf{right}_t := t + 1\) if \(t < T\), \(\texttt{null}\) else.
    \State \(\textbf{assigned}_t := \text{false}\)
\EndFor
\For{event $E = (e, t) \in P$} %
    \State \(t' := \find(t)\)
    \State \(d_\text{left} := t - \textbf{left}_{t'}\)
    \State \(d_\text{right} := \textbf{right}_{t'} - t\)
    \State  \(\text{goleft} := \textbf{true}\) with probability \(\frac{1/d_\text{left}}{1/d_\text{left} + 1/d_\text{right}}\), \textbf{false} else.
    \If{\text{goleft}}
        \State \(\textbf{assigned}_{\textbf{left}_{t'}} := \text{true}\)
        \If{\(\textbf{assigned}_{(\textbf{left}_{t'} - 1)}\)}
            \State \(\union(\textbf{left}_{t'} - 1, \textbf{left}_{t'})\)
            \State \(\union(\textbf{left}_{t'}, t)\)
            \State \(t'' := \find(t)\)
            \State \(\textbf{left}_{t''} := \textbf{left}_{(\textbf{left}_{t'} - 1)}\)
            \State \(\textbf{right}_{t''} := \textbf{right}_{t'}\)
        \Else
            \State \(\union(\textbf{left}_{t'}, t)\)
            \State \(t'' := \find(t)\)
            \State \(\textbf{left}_{t''} := \textbf{left}_{t'} - 1\)
            \State \(\textbf{right}_{t''} := \textbf{right}_{t'}\)
        \EndIf
    \Else 
        \If{\(\textbf{assigned}_{(\textbf{right}_{t'} + 1)}\)}
            \State \(\union(\textbf{right}_{t'} + 1, \textbf{right}_{t'})\)
            \State \(\union(\textbf{right}_{t'}, t)\)
            \State \(t'' := \find(t)\)
            \State \(\textbf{left}_{t''} := \textbf{left}_{t'}\)
            \State \(\textbf{right}_{t''} := \textbf{right}_{(\textbf{right}_{t'} + 1)}\)
        \Else
            \State \(\union(\textbf{right}_{t'}, t)\)
            \State \(t'' := \find(t)\)
            \State \(\textbf{left}_{t''} := \textbf{left}_{t'}\)
            \State \(\textbf{right}_{t''} := \textbf{right}_{t'} + 1\)
        \EndIf
    \EndIf
\EndFor
\end{algorithmic}
\end{algorithm}

\partitiontreedepth*

\begin{proof}
    Consider an arbitrary day \(t \in [T]\) which is a leaf in the partition-tree as defined in~\cref{def:random-partition-tree}. 
    We analyze the depth of the leaf containing \(t\) in the partition-tree for every $t \in [T]$.  
    Consider the path from the root of the partition-tree to this leaf, that is the path of all windows containing \(t\).  

    Consider some non-leaf window \(W\) of size \(S\) that contains \(t\); the size of window $W$ is a random variable with value $S$ conditioned
    on a set of events.  We know that with probability at least \(\frac{1}{2}\), the random partition of \(W\) results in both children of \(W\) having size at most \(\frac{3S}{4}\).  We call this the ``good event."  (Consider the center \(\lfloor \frac{S}{2}\rfloor\) elements of \(W\).  There are strictly more dividers that border one of these center elements than there are dividers that do not border these elements.  If any of these dividers is chosen, the smaller child has size at least \(\frac{S}{4}\), thus both children have size at most \(\frac{3S}{4}\).)

    Good events can occur at most \(\log_{\frac{4}{3}}(T) = \frac{\ln(T)}{\ln \frac{4}{3}} < 4 \ln (T)\) times over the windows on the root-to-leaf path to $t$ 
    before there is one element left and the path ends.  

    Now we bound the probability that \(t\)'s path in the partition-tree has length more than \(k \ln (T)\), for some constant \(k \ge 36\).  For this to occur, it is necessary that in the partitioning for the first \(k \ln(T)\) windows on \(t\)'s path, the good event occurred fewer than \(4 \ln (T)\) times.  

    We can bound the probability that good events occur fewer than $4 \ln(T)$ times with the Chernoff-Hoeffding inequality. Furthermore, for this 
    analysis, we assume that the number of windows along $t$'s path is at least $k\ln(T)$; otherwise, we have our desired property.
    Let \[X = X_1 + \dots + X_{k \ln(T)},\] 
    where \(X_j\) is the indicator random variable that the good event occurred at the \(j\)th window of \(i\)'s path.  We have \(\mathbf{E}[X_j] \ge \frac{1}{2}\) and \(\mathbf{E}[X] \ge \frac{k\ln (T)}{2}\).  Furthermore, each good event is independent of previous good events. Thus,
    \begin{align*}
        \textbf{Pr} \left[ X \le \mathbf{E}[X] - t \right] &\le e^{-\frac{t^2}{3 \mathbf{E}[X]}} \\
        \textbf{Pr} \left[X \le 4 \ln(T) \right] &= \textbf{Pr} \left[X \le \mathbf{E}[X] - (\frac{k}{2} - 4) \ln(T) \right]\\
        &\le e^{- \frac{((\frac{k}{2} - 4) \ln(T))^2}{\frac{3k\ln(T)}{2}}} \\
        &= e^{- (\frac{k}{6} - \frac{8}{3} + \frac{32}{3k}) \ln(T)} \\
        &\le e^{- \frac{k}{12} \ln T} &\text{for }k \ge 28\\
        &= T^{-\frac{k}{12}}.
    \end{align*}
    This shows that any fixed day \(t\) has depth at most \(k \ln T\) in the random partition-tree, with probability at least \(1 - T^{-\frac{k}{12}}\).  

    Now, we can take a union bound over the \(T\) days, and say that all \(T\) days have depth at most \(k \ln T\) in the random partition-tree with probability  
    \[\ge 1 - T \cdot T^{-\frac{k}{12}} = 1 - T^{-\frac{k}{12} + 1} \ge 1 - T^{-\frac{k}{24}}.\]

    To bound the expected depth of the tree, 
    we can multiply the probability of the depth being low by the depth, 
    and, otherwise, multiply the maximum possible depth of \(T\) by the remaining probability.  Using the above bound for \(k = 96\), we get  
    \[\mathbf{E}[\mathrm{depth}] \le 96 \ln T \cdot (1 - T^{-4}) + T \cdot T^{-4} = O(\log(T)).\]
\end{proof}

\maintainschedule*

\begin{proof}
    \Cref{lem:preprocessing-schedule} guarantees that at the end of preprocessing, we have at most two events scheduled per day.  Call this assignment of events \(\mathbf{a}_0\).  Over the course of the algorithm, events are rescheduled via a doubling search procedure using \textsc{ProcessEventLaterThanPrediction}.  For an event \(e\), call this sequence of reschedulings \(\mathbf{a}_1(e), \mathbf{a}_2(e), \dots\), where \(\mathbf{a}_i(e) = \mathbf{a}_0(e) + 2^i\), if it exists. 
    Thus the set of existent \(\mathbf{a}_i(e)\), for event \(e\), is some subset of the events in \(\mathbf{a}_0\), shifted by exactly \(2^i\).  Thus, each \(\mathbf{a}_i\) contributes at most a constant number of scheduled events on each day.  

    Now, we note that since there are only \(T\) possible days, an event can only be rescheduled \(O(\log T)\) times.  Thus, \(\mathbf{a}_i(e)\) only exist for \(i \in [k \cdot \log T)]\) for some constant $k \geq 1$.  
    Together, this results in at most \(O(\log T)\) events (re)scheduled for any given day. There are $T$ days, resulting in
    $O(T \log T)$ reschedules.

    Now, consider the work done by our algorithm to maintain the schedule on a given day \(t\).  The algorithm must potentially reschedule the prediction of one real event for day \(t\) using \textsc{ProcessEventEarlierThanPrediction}, and reschedule (via the guess and double procedure in \textsc{ProcessEventLaterThanPrediction}) the predicted events that were scheduled for day \(t\) but did not occur.  
    Rescheduling a predicted event using \textsc{ProcessEventEarlierThanPrediction} is done once per real event; since there are $T$ real events,
    this procedure requires $O(T)$ reschedules.
    By our previous accounting, there are at most \(O(T\log T)\) events that are rescheduled by \textsc{ProcessEventLaterThanPrediction}.  
    Hence, our algorithm does $O(T\log T)$ reschedules in total.
\end{proof}

\subsection{Best-of-All-Worlds Backstop Proofs}\label{app:backstop}

\backstop*
\begin{proof}
    At a high level, our meta-algorithm will update all of the algorithms in synchrony, but it will perform computations at the rate of the faster algorithm.  Thus, the faster algorithm will be up-to-date with the current timestep and will provide the output computation; each slower algorithm may have a backlog of computations that it has not yet performed.  This is described in~\cref{alg:backstop}.

    We fix a day \(t\) and analyze the work done by the meta-algorithm up through the \(t\)-th day.~\cref{algline:add-event-to-buffer} does \(O(t \cdot N)\) work over \(t\) updates.  We can assume that \(R_{A_i}(t) = \Omega(t)\) for all of the algorithms \(A_i\) (where the algorithms perform \(\Omega(1)\) processing per update), and thus the contribution of \(O(t \cdot N)\) to the work will be dominated by faster growing terms.  
    
    \cref{algline:interleave-computation} maintains the invariant that each of the algorithms \(A_1, \dots, A_N\) has been run for the same number of computation steps.  Thus, the first algorithm to complete the computation for all of events in its buffer after the \(t\)-th update will be precisely the algorithm \(A_i\) that minimizes \(R_{A_i}(t)\).  Over the course of all updates up to \(t\), the meta algorithm has run \(A_i\) for 
    \[R_{A_i}(t) = \min\{R_{A_1}(t), \dots, R_{A_N}(t)\} \text{ steps.}\] 
    Since it has run each of the \(N\) algorithms for the same number of steps, it has done total work 
    \[O\left( N \cdot \min\{R_{A_1}(t), \dots, R_{A_N}(t)\} \right).\]

    Correctness of the algorithm follows because a solution output on day \(t\) is the solution output by $A_i$ upon seeing exactly the stream of events up to day \(t\).  Thus, if $A_i$ is correct, the meta-algorithm is also correct.
\end{proof}

\subsection{Additional Related Work and Connections}
\label{app:related-work}

\paragraph{Sliding window and look-ahead models.}  In the \emph{sliding window model}, introduced by \cite{DGIM02}, the algorithm views an infinite stream of data, and must maintain a statistic over the last \(N\) data points seen (where \(N\) is the width of the window).  The sliding window model has a large body of work in the streaming literature \cite{PGD15, WLLSDW16, ELVZ17, BDMMUWZ20, EMMZ22, JWZ22, WZ22}, including \cite{CMS13, Rei19, BdBM21, ADKN23} who give \emph{semi-streaming} algorithms for graph problems in this model.  Interestingly, despite being an inherently dynamic model, (consisting of insertions, and deletions exactly \(N\) days after insertion), it has not been studied much in the dynamic algorithms literature.  
A related family of models are \emph{look-ahead models}, where a dynamic algorithm is given information about future events.  Examples include graph algorithms with access to the set of vertices involved in the next few updates (but not the full sequence of edge operations) \cite{KMW96}, and algorithms that have full access to the next few operations \cite{SM10}.  

The model that is most relevant to our work is the \emph{deletion look-ahead model}.  In this model, the algorithm maintains a statistic over a subset of some ground set of elements (e.g. a graph is a subset of possible edges), and has access to the future deletion times of all existing elements.  This is a strict generalization of the sliding window model.  It is known that designing an algorithm with an amortized guarantee in this model can be reduced to designing an algorithm with a worst-case guarantee in the \emph{incremental} model, where elements are only inserted and never deleted \cite{chan2011three, vdBNS19}.  This reduction and a stronger reduction that achieves a worst-case guarantee in the known-deletion model, are formalized in \cite{PR23}. \cite{PR23} actually considers a slightly more general setting in which the algorithm has access to the order in which the elements will be deleted, and not necessarily the exact deletion times.  Our main contribution considers a more general model, in which information about the future is subject to error.  We note that the \emph{predicted-deletion dynamic model} that we introduce is a strict generalization of the known-deletion model, which is itself a strict generalization of the sliding window model.

\paragraph{Algorithms with predictions. }
\emph{Algorithms with predictions}, also often called \emph{learning augmented algorithms}, is a paradigm that has been gaining much attention in recent years.  
An algorithm solicits \emph{predictions} from an untrusted source (e.g.\ a machine learning model) to help make decisions.  The goal is to design algorithms that achieve the following three desiderata \cite{LV21}: 
\begin{enumerate}[(1)]
    \item \textbf{(Consistency)} If the predictions are of high quality, the algorithm performs much better than a worst-case algorithm. 
    \item \textbf{(Competitiveness)} If the predictions are of low quality, the algorithm does not perform any worse than a worst-case algorithm. 
    \item \textbf{(Robustness)} The performance of the algorithm degrades gracefully as a function of the prediction error.
\end{enumerate}
Additionally, we want to solicit predictions that can be reasonably obtained in practice.  A detailed overview of the field is given in ~\cite{BWCA-AlgosWithPredictions}.

These desiderata present a challenge. While we would like to achieve all three for some reasonable notion of prediction, it is not a priori clear that such a guarantee is even possible.  For some problems and settings, it is not indeed possible to achieve all three simultaneously, and algorithms are designed that allow the user to trade off these objectives.  One example is the classic online problem of rent-or-buy, which exhibits an inherent trade-off between consistency and robustness \cite{KPS18, GP19, WZ20}.  Thus it is particularly interesting that it is indeed possible to achieve all three points for dynamic algorithms in our setting.  

The setting closest to the dynamic model is the \emph{warm start} setting.  In this setting, a static algorithm is given an instance, along with some additional predicted information.  Examples include graph problems \cite{DILMV21, CSVZ22, DMVW23}, in which the prediction is a candidate solution, and the quality of the prediction is measured as the distance from the prediction to the true optimal solution.  
In the warm start setting, it is often the case that algorithms are able to achieve consistency and robustness simultaneously.  A strong motivation for this model is time-series data, in which we want to solve a series of instances with the context that ``yesterday's instance is likely not too different from today's instance."  In this setting, we can think of using each day's solution as the prediction for the next day's instance as a kind of dynamic procedure.  The guarantees for such a process differ from the standard dynamic model.  In the dynamic model, the work done by the algorithm on a given day scales with the magnitude of the change in the instance, whereas in the warm start setting, the work scales with the magnitude of the change in the solution.  In general, results of these types are incomparable.  One benefit of the warm start setting is that it models a larger range of possible updates between subsequent instances.  On the other hand, we cannot expect the same kinds of update times that we get for fully dynamic algorithms, since for many of these problems, even checking if a predicted solution is optimal for a fresh instance can take time linear in the size of the instance.  Thus our work differs significantly from work on warm starts, both in terms of the kinds of predictions we expect, and in the kinds of guarantees we achieve.  

Online algorithms are the area in which algorithms with predictions were first studied, starting with the motivating work of~\cite{KBCDP18}, which demonstrated that machine learned predictions could dramatically improve the efficiency of index structures, both in theory and in practice.  Many online problems have since been studied in this model.  Some, including rent-or-buy \cite{KPS18, GP19, WZ20}, and problems related to caching \cite{LV21, BMRS20, JPS20, Roh20, Wei20,  ADJKR20, IMMR22}, utilize predictions of when future events occur.  Others such as the secretary problem \cite{DLPV21} and combinatorial optimization problems \cite{BMS20, LMRX21}, use predictions of what the optimal solution is.  An important line of work includes scheduling and queuing problems \cite{ACEPS20, LLMV20, Mit21}, in which access to predictions about parameters such as job length can allow an algorithm to circumvent strong lower bounds.  
This is related to our work, as our algorithm essentially schedules partial computations with access to predictions of the longevity of the edges involved.  In the online setting, it is often the case that problems exhibit an inherent tradeoff between consistency and robustness, and algorithms often include a tunable parameter that allows the user to trade off these objectives.

Some sequential settings require predictions of the frequencies of certain objects in a sequence, such as \cite{CGP20, JLLRW20, CEILMRSWWZ22} which design streaming algorithms, and \cite{EINRSW21} which minimizes sample complexity.  Other settings include learning good hashes and sketches of data \cite{Mit18, VKMK21, IVY19, HIKV18}, and using predicted information to maximize revenue in mechanism design settings \cite{MMV17}.

\subsection{Deterministic \(\ell_\infty\)-based bound}
\label{sec:deterministic-l-inf-bound}

In this section we show that the greedy allocation algorithm has competitive ratio bounded by a function of the minimum \(\ell_\infty\) norm of any solution.

Let the maximum error between any $\trep$ and $\tdel$ be denoted as $\maxe$. We show that our greedy strategy of assigning
each element update event prediction 
to the nearest free time slot obtains total error at most $t \cdot \maxe$ where $t$ is the total number of updates. 
Each slot can be taken by an event at most $\maxe$ away time-steps away. Thus, this means that given any two time slots $t_2$
and $t_1$ where $t_2 > t_1$, it holds that the number of events in $[t_1, t_2]$ is at most $t_2 - t_1 + 1 + 2\cdot \maxe$.
Using this observation, we prove the following lemma about the maximum error by our assignment of events to the timeline.

\greedymatchingguarantee*

\begin{proof}
    We first prove the observation that between any two time slots $t_2$ and $t_1$, the maximum number of updates with predicted
    times in $[t_1, t_2]$ is $t_2 - t_1 + 1 + 2 \cdot \maxe$. This is the case because in the true update sequence, the at any time
    slot $t'$ at most $\maxe$ updates on the left of $t'$ can have $\trep$ equal to $t'$ and at most $\maxe$ updates on the 
    right of $t'$ can have $\trep$ equal to $t'$. Now, we make the following potential argument. Give each update 
    $4\maxe \cdot \log(T)$ 
    coins when we assign the update to a time slot. 
    Each time we greedily assign a new update to the nearest free time slot, we take a coin from each update 
    that we \emph{pass on the way to the new assigned time slot}. We then need to argue that no update runs out of coins.

    For any time slot $t'$, an update passes $t'$ if the shortest distance to a free slot includes $t'$. We consider
    two contiguous segments of time of size $2^i$ for all $i \in [\log_2(T)]$
    where $t'$ is on the boundary. In other words, we consider $[t'- 2^i - 1, t']$ and $[t', t' + 2^i + 1]$.
    We consider such contiguous segments of time because only additional updates in the range $[t' - 2^{i-1}, t']$  (symmetrically
    $[t', t'+2^{i-1}]$) will cross $t'$ if there are no free slots in $[t'- 2^i - 1, t'-2^{i-1}]$ (symmetrically
    $[t', t' + 2^{i-1}]$. Otherwise, if there are free slots in those segments, then the updates
    will fill up those free slots first. By our argument above,
    for any such contiguous segment of time, there are at most $2\maxe$ additional updates that must be assigned. Hence,
    in the contiguous segment of time $[t' - 2^i - 1, t']$ at most $2\maxe$ updates are assigned and these additional updates may
    pass $t'$. The same holds for the contiguous segment of time $[t', t' + 2^i + 1]$. 
    Thus, a total of at most $4\maxe$ updates pass
    through $t'$ for each $i \in [\log_2(T)]$ and we do not use more than $4\maxe \cdot \log_2(T)$ coins from $t'$. 
    The additional $\maxe$ error is for the initial error of the prediction. 
\end{proof}

\bibliographystyle{alpha}
\bibliography{ref}
\end{document}